\documentclass[a4paper,USenglish,cleveref, autoref, thm-restate]{lipics-v2021}
\usepackage[T1]{fontenc}
\usepackage[utf8]{inputenc}
\usepackage{numprint}

\usepackage[T1]{fontenc}
\usepackage[utf8]{inputenc}
\usepackage{numprint}
\usepackage{graphicx}
\usepackage{caption}
\usepackage{amsmath}
\usepackage{stackrel}
\usepackage{amsthm}
\usepackage{mathrsfs}
\usepackage{amssymb}
\usepackage{amsfonts}
\usepackage{wasysym}
\usepackage{stmaryrd}
\usepackage{tikz}
\usepackage[ruled, vlined, linesnumbered]{algorithm2e}
\usepackage{stmaryrd}
\usepackage{enumitem}
\usetikzlibrary{automata}
\usetikzlibrary{shapes}
\usepackage{hyperref}
\usepackage{comment}

\newcommand{\<}{\langle}
\renewcommand{\>}{\rangle}
\newcommand{\lv}{\lVert}
\newcommand{\rv}{\rVert}

\newcommand{\bsigma}{\bar{\sigma}}
\newcommand{\btau}{\bar{\tau}}

\renewcommand{\|}{\upharpoonright}

\newcommand{\N}{\mathbb{N}}
\newcommand{\Z}{\mathbb{Z}}
\newcommand{\R}{\mathbb{R}}
\newcommand{\Q}{\mathbb{Q}}
\newcommand{\bR}{\overline{\mathbb{R}}}
\renewcommand{\S}{\mathbb{S}}

\newcommand{\dseal}{\!\,^\llcorner}

\newcommand{\playcircle}{{\scriptsize{\Circle}}}

\renewcommand{\l}{\ell}
\renewcommand{\epsilon}{\varepsilon}
\renewcommand{\phi}{\varphi}

\newcommand{\card}{\mathsf{card}}
\newcommand{\SC}{\mathsf{SC}}
\newcommand{\Plays}{\mathsf{Plays}}
\newcommand{\Hist}{\mathsf{Hist}}

\newcommand{\Conv}{\mathsf{Conv}}

\newcommand{\nego}{\mathsf{nego}}
\newcommand{\lCons}{\lambda\mathsf{Cons}}
\newcommand{\lRat}{\lambda\mathsf{Rat}}

\newcommand{\ML}{\mathsf{ML}}

\newcommand{\Abs}{\mathsf{Abs}}
\newcommand{\Conc}{\mathsf{Conc}}

\newcommand{\val}{\mathsf{val}}
\newcommand{\Sol}{\mathsf{Sol}}
\newcommand{\Red}{\mathsf{Red}}
\newcommand{\last}{\mathsf{last}}
\newcommand{\first}{\mathsf{first}}
\newcommand{\Occ}{\mathsf{Occ}}
\newcommand{\Inf}{\mathsf{Inf}}
\newcommand{\MP}{\mathsf{MP}}

\renewcommand{\P}{\mathbb{P}}
\newcommand{\C}{\mathbb{C}}
\renewcommand{\S}{\mathbb{S}}

\newcommand{\NP}{\mathbf{NP}}
\newcommand{\ExpTime}{\mathbf{ExpTime}}
\newcommand{\PTime}{\mathbf{P}}

\newcommand{\bx}{\bar{x}}
\newcommand{\by}{\bar{y}}
\newcommand{\bz}{\bar{z}}

\newcommand{\bbx}{\bar{\bar{x}}}

\newcommand{\bbalpha}{\bar{\bar{\alpha}}}
\newcommand{\ba}{\bar{a}}

\newcommand{\bzero}{\bar{0}}

\newcommand{\dH}{\dot{H}}
\newcommand{\dC}{\dot{C}}
\newcommand{\dc}{\dot{c}}
\newcommand{\dpi}{\dot{\pi}}
\newcommand{\dchi}{\dot{\chi}}
\newcommand{\dxi}{\dot{\xi}}

\theoremstyle{plain}
\newtheorem{thm}{Theorem}

\newtheorem{lm}{Lemma}
\newtheorem{cor}{Corollary}

\newtheorem{innercustomthm}{Theorem}
\newenvironment{customthm}[1]
  {\renewcommand\theinnercustomthm{#1}\innercustomthm}
  {\endinnercustomthm}
\newtheorem{innercustomlem}{Lemma}
\newenvironment{customlem}[1]
  {\renewcommand\theinnercustomlem{#1}\innercustomlem}
  {\endinnercustomlem}

\theoremstyle{definition}
\newtheorem{defi}{Definition}

\theoremstyle{remark}
\newtheorem*{rk}{Remark}

\newtheorem{exa}{Example}

\title{The Complexity of SPEs in Mean-payoff Games}

\author{Léonard Brice}
{
Université Libre de Bruxelles, Belgium
}{
leonard.brice@ulb.be
}{}{}

\author{Jean-François Raskin}{Université Libre de Bruxelles, Belgium}{jraskin@ulb.ac.be}{}{}

\author{Marie van den Bogaard}{Univ Gustave Eiffel, CNRS, LIGM, F-77454 Marne-la-Vallée, France}{marie.van-den-bogaard@univ-eiffel.fr}{}{}

\authorrunning{
L. Brice, J.-F. Raskin, and M. van den Bogaard
}

\Copyright{Léonard Brice, Jean-François Raskin, Marie van den Bogaard}



\funding{This work is partially supported by the ARC project Non-Zero Sum Game Graphs: Applications to Reactive Synthesis and Beyond (F\'{e}d\'{e}ration Wallonie-Bruxelles), the EOS project Verifying Learning Artificial Intelligence Systems (F.R.S.-FNRS \& FWO), the COST Action 16228 GAMENET (European Cooperation in Science and Technology), and by the PDR project \emph{Subgame perfection in graph games} (F.R.S.- FNRS).}

\nolinenumbers

\begin{document}

\maketitle

\begin{abstract}
\noindent
We establish that the subgame perfect equilibrium (SPE) threshold problem for mean-payoff games is $\NP$-complete. While the SPE threshold problem was recently shown to be decidable (in doubly exponential time) and $\NP$-hard, its exact worst case complexity was left open. 
%Our starting point is the recent characterization of SPEs that establishes decidabilty and show that SPE outcomes in mean-payoff games are exactly the plays that are consistent with the fixed points of the \emph{negotiation function} that characterizes worst-case value against rational adversaries. We provide a bound on the size of its least fixed point, and using that result, we draw links between the negotiation function and a new class of games, the \emph{reduced negotiation games}.
%We provide a nondeterministic polynomial time algorithm to solve those games.
\end{abstract}

\section{Introduction}
Nash equilibria (NEs), a fundamental solution concept in game theory,  are defined as strategy profiles such that no player can improve their payoff by changing unilaterally their strategy. So NEs can be interpreted as self-inforcing contracts from which there is no incentive to deviate unilaterally. Unfortunately, NEs are known to suffer, in sequential games like infinite duration games played on graphs, from the issue of \emph{non-credible threats}: to enforce a NE, some players may threaten other players to play irrationally in order to punish deviations. This is allowed by the definition of NEs, as in case of deviation from one player, the other players are not bound to rational behaviors anymore and they can therefore play irrationally w.r.t. their own objectives in order to sanction the deviating player.
This drawback of NEs has triggered the introduction of the notion of \emph{subgame-perfect equilibria} (SPEs)~\cite{Osborne}, a more complex but more natural solution concept for sequential games.
A strategy profile is an SPE if after every history, i.e. in every \emph{subgame}, the strategies of the players still form an NE.
Thus, SPEs impose rationality even after a deviation and only rational behaviors can be used to coerce the behavior of other players.
%in all what the players \emph{plan} to do, in some scenario, and not only in what they effectively do.

In this paper, we study the complexity of SPE problems in infinite-duration sequential games played on graphs with mean-payoff objectives.
While NEs always exist in those games, as proved in~\cite{DBLP:conf/lfcs/BrihayePS13}, SPEs do not always exist as shown in~\cite{solan2003deterministic,DBLP:conf/csl/BrihayeBMR15}.
The \emph{SPE threshold problem}, i.e. the problem of deciding whether a given mean-payoff game admits an SPE satisfying some constraints on the payoffs it grants to the players, has recently been proved to be decidable in \cite{Concur}.
However, its worst-case computational complexity is open: \cite{Concur}~provides only an $\NP$ lower bound and a $2\ExpTime$ upper bound.
In this paper, we close this complexity gap and prove that the problem is actually $\NP$-complete.

\subparagraph{{\bf Contributions.}}
The starting point of our algorithm is the characterization of SPEs recently presented in \cite{Concur}, based on the notions of \emph{requirement} and \emph{negotiation function}.
A requirement on a game $G$ is a function $\lambda: V \to \R \cup \{\pm \infty\}$, where $V$ is the state space of $G$.
For a given state $v$, the value $\lambda(v)$ should be understood as the minimal payoff that the player controlling the state $v$ will require in a play traversing $v$ in order to avoid deviating.
A requirement captures, therefore, some level of \emph{rationality} of the players.
The negotiation function transforms each requirement $\lambda$ into a (possibly stronger) requirement $\nego(\lambda)$, such that $\nego(\lambda)(v)$ is the best payoff that the player controlling $v$ can ensure, while playing against a coalition of the other players that play rationally with regards to the requirement $\lambda$.
A play is the outcome of an SPE if and only if it satisfies the requirement $\lambda^*$, the least fixed point of the negotiation function --- or equivalently, one of its fixed points.
We recall that result in Lemma~\ref{lm_spe}.
In order to obtain our nondeterministic polynomial time algorithm, the rest of the paper constructs a notion of witness recognizing the positive instances of the SPE threshold problem.
Such witnesses admit three pieces.
First, we show that the size of $\lambda^*$ can be bounded by a polynomial function of the size of the game (Theorem~\ref{thm_size_lambda}).
This result is obtained by showing that the set of fixed points of the negotiation function can be characterized by a finite union of polyhedra that in turn can be represented by linear inequations.
While the number of inequations that are needed for that characterization may be large (it cannot be bounded polynomially), we show that each of those inequations have coefficients and constants whose binary representations can be bounded polynomially.
As the least fixed point is the minimal value in this set, it is represented by a vertex of one of those polyhedra.
Then this guarantees, using results that bounds the solutions of linear equalities, that the least fixed point has a binary representation that is polynomial and so it can be guessed in polynomial time by a nondeterministic algorithm: it will be the first piece of our notion of witness, in the non-deterministice algorithm we design to solve the SPE threshold problem.
Second, we define a witness of polynomial size for the existence of a play, consistent with a given requirement, which generates a payoff vector between the desired thresholds (Theorem~\ref{thm_constrained_existence_lambda_cons}). This play is not guaranteed to be regular.
Third, we define a witness of polynomial size to prove that a requirement is indeed a fixed point of the negotiation function.
This notion of certificate relies on a new and more compact game characterization of the negotiation function called the {\em reduced negotiation game} (Definition~\ref{def_reduced_game}, Theorem~\ref{thm_reduced_game}).
These results are far from trivial as we also show that SPEs may rely on strategy profiles that are not regular and require infinite memory. As both the least fixed point and its two certificates can be guessed and verified in polynomial time, we obtain $\NP$ membership for the threshold problem, closing the complexity gap left open in~\cite{Concur} (Theorem~\ref{thm_np_complete}).

Additionally, all the previous results do also apply to $\epsilon$-SPEs, a quantitative relaxation of SPEs.
In particular, Theorem~\ref{thm_np_complete} does also apply to the $\epsilon$-SPE threshold problem.

\subparagraph{{\bf Related works.}}
Non-zero sum infinite duration games have attracted a large attention in recent years, with applications targeting reactive synthesis problems, see e.g.~\cite{DBLP:conf/lata/BrenguierCHPRRS16,DBLP:conf/dlt/Bruyere17,DBLP:journals/siglog/Bruyere21, DBLP:conf/tacas/FismanKL10, DBLP:journals/amai/KupfermanPV16} and their references.
We now detail other works more closely related to our contributions.

In~\cite{DBLP:conf/csl/BrihayeBMR15}, Brihaye et al. introduce and study the notion of weak SPE, which is a weakening of the classical notion of SPE. This weakening is equivalent to the original SPE concept on reward functions that are {\em continuous}.
This is the case for example for the quantitative reachability reward function, on which Brihaye et al. solve the SPE threshold problem in \cite{DBLP:conf/concur/BrihayeBGRB19}.
The mean-payoff cost function is not continuous and the techniques used in~\cite{DBLP:conf/csl/BrihayeBMR15}, and generalized in~\cite{DBLP:conf/fossacs/Bruyere0PR17}, cannot be used to characterize SPEs for the mean-payoff reward function.

In~\cite{thesis_noemie}, Meunier develops a method based on Prover-Challenger games to solve the problem of the existence of SPEs on games with a finite number of possible payoffs. In mean-payoff games, the number of possible payoffs is uncountably infinite.

In~\cite{DBLP:journals/mor/FleschP17}, Flesch and Predtetchinski present another characterization of SPEs on games with finitely many possible payoffs, based on a game structure with infinite state space.
%to which Brice, Raskin and van den Bogaard have later given the name of \emph{abstract negotiation game}.
In~\cite{Concur}, Brice et al. define the notions of requirements and negotiation function. They prove that the negotiation function is characterized by a zero-sum two-player game called {\em abstract negotiation game}, which is similar to the game introduced in the characterization of Flesch and Predtetchinski. As a starting point for algorithms, they also provide an effective representation of this game, called \emph{concrete negotiation game}, which turns out to be a zero-sum finite state multi-mean-payoff games~\cite{DBLP:journals/iandc/VelnerC0HRR15}.
Finally, they use those tools to prove that the SPE threshold problem is decidable for mean-payoff games.
They left open the question of its precise complexity: they provide a $\NP$ lower bound and a $2\ExpTime$ upper bound.
In~\cite{CSL}, the same authors use those tools to close the complexity gap for the SPE threshold problem in parity games, which had been proved to be $\ExpTime$-easy and $\NP$-hard by Ummels and Grädel in \cite{GU08}.
They prove that the problem is actually $\NP$-complete.
The techniques used in that paper heavily rely on the fact that parity objectives are $\omega$-regular, which is not the case of mean-payoff games in general.

In \cite{DBLP:conf/concur/ChatterjeeDEHR10}, Chatterjee et al. study mean-payoff \emph{automata}, and give a result that can be translated into an expression of all the possible payoff vectors in a mean-payoff game.
In \cite{DBLP:conf/cav/BrenguierR15}, Brenguier and Raskin give an algorithm to build the Pareto curve of a multi-dimensional two-player zero-sum mean-payoff game.
To do so, they study systems of equations and of inequations, and they prove that they always admit simple solutions (with polynomial size).
Those technical results will be used along this paper.

\subparagraph{{\bf Structure of the paper.}}
In Section~2, we introduce the necessary background.
Section~3 recalls the notions of requirement and negotiation function, and link them to NEs and SPEs.
Section~4 recalls results about the size of solutions of systems of equations or inequations, and use them to bound the size of the least fixed point of the negotiation function.
Section~5 defines a witness for the existence of a $\lambda$-consistent play between two given thresholds.
Section~6 introduces the reduced negotiation game that is a new compact characterization of the negotiation function.
Finally, Section~7 applies those results to prove the $\NP$-completeness of the SPE threshold problem on mean-payoff games.
Due to the page limit, some proofs are not given in the main body of the text: they appear in a well-identified appendix.

	\section{Background} \label{sec_background}

\subparagraph{{\bf Games, strategies, equilibria.}}
In all what follows, we study infinite duration turn-based quantitative games on finite graphs with complete information.

\begin{defi}[Game]\label{defi_game}
	A \emph{non-initialized game} is a tuple
	$G = \left(\Pi, V, (V_i)_{i \in \Pi}, E, \mu\right)$, where:
	\begin{itemize}
		\item $\Pi$ is a finite set of \emph{players};
		
		\item $(V, E)$ is a directed graph, called the \emph{underlying graph} of $G$, whose vertices are sometimes called \emph{states} and whose edges are sometimes called \emph{transitions}, and in which every state has at least one outgoing transition.
		For the simplicity of writing, a transition $(v, w) \in E$ will often be written $vw$;
		
		\item $(V_i)_{i \in \Pi}$ is a partition of $V$, in which $V_i$ is the set of states \emph{controlled} by player $i$;
		
		\item $\mu: V^\omega \to \R^\Pi$ is an \emph{payoff function}, that maps each infinite word $\rho$ to the tuple $\mu(\rho) = (\mu_i(\rho))_{i \in \Pi}$ of the players' \emph{payoffs}.
	\end{itemize}
	\noindent
	An \emph{initialized game} is a tuple $(G, v_0)$, often written $G_{\|v_0}$, where $G$ is a non-initialized game and $v_0 \in V$ is a state called \emph{initial state}.
	We often use the word \emph{game}, alone, for both initialized and non-initialized games.
\end{defi}

\begin{defi}[Play, history]
	A \emph{play} (resp. history) in the game $G$ is an infinite (resp. finite) path in the graph $(V, E)$.
	It is also a play (resp. history) in the initialized game $G_{\|v_0}$, when $v_0$ is its first vertex.
	The set of plays (resp. histories) in the game $G$ (resp. the initialized game $G_{\|v_0}$) is denoted by $\Plays G$ (resp. $\Plays G_{\|v_0}, \Hist G, \Hist G_{\|v_0}$).
	We write $\Hist_i G$ (resp. $\Hist_i G_{\|v_0}$) for the set of histories in $G$ (resp. $G_{\|v_0}$) of the form $hv$, where $v$ is a vertex controlled by player $i$.
	
	Given a play $\rho$ (resp. a history $h$), we write $\Occ(\rho)$ (resp. $\Occ(h)$) the set of vertices that appear in $\rho$ (resp. $h$), and $\Inf(\rho)$ the set of vertices that appear infinitely often in $\rho$.
For a given index $k$, we write $\rho_{\leq k}$ (resp. $h_{\leq k}$), or $\rho_{< k+1}$ (resp. $h_{<k+1}$), the finite prefix $\rho_0 \dots \rho_k$ (resp. $h_0 \dots h_k$), and $\rho_{\geq k}$ (resp. $h_{\geq k}$), or $\rho_{> k-1}$ (resp. $h_{>k-1}$), the infinite (resp. finite) suffix $\rho_k \rho_{k+1} \dots$ (resp. $h_k h_{k+1} \dots h_{|h|-1}$).
Finally, we write $\first(\rho)$ (resp. $\first(h)$) the first vertex of $\rho$ (and $\last(h)$ the last vertex of $h$).
\end{defi}

\begin{defi}[Strategy, strategy profile]
	A \emph{strategy} for player $i$ in the initialized game $G_{\|v_0}$ is a function $\sigma_i: \Hist_i G_{\|v_0} \to V$, such that $v\sigma_i(hv)$ is an edge of $(V, E)$ for every $hv$.
	A history $h$ is \emph{compatible} with a strategy $\sigma_i$ if and only if $h_{k+1} = \sigma_i(h_0 \dots h_k)$ for all $k$ such that $h_k \in V_i$. A play $\rho$ is compatible with $\sigma_i$ if all its prefixes are.

	A \emph{strategy profile} for $P \subseteq \Pi$ is a tuple $\bsigma_P = (\sigma_i)_{i \in P}$, where each $\sigma_i$ is a strategy for player $i$ in $G_{\|v_0}$.
    A play or a history is \emph{compatible} with $\bsigma_P$ if it is compatible with every $\sigma_i$ for $i \in P$.
    A \emph{complete} strategy profile, usually written $\bsigma$, is a strategy profile for $\Pi$.
    Exactly one play is compatible with a complete strategy profile: we write it $\< \bsigma \>$, and call it the \emph{outcome} of $\bsigma$.
    
    When $i$ is a player and when the context is clear, we will often write $-i$ for the set $\Pi \setminus \{i\}$.
    When $\btau_P$ and $\btau'_Q$ are two strategy profiles with $P \cap Q = \emptyset$, we write $(\btau_P, \btau'_Q)$ the strategy profile $\bsigma_{P \cup Q}$ such that $\sigma_i = \tau_i$ for $i \in P$, and $\sigma_i = \tau'_i$ for $i \in Q$.
\end{defi}

Before moving on to SPEs, let us recall that an NE is a strategy profile such that no player has an incentive to deviate unilaterally.

\begin{defi}[Nash equilibrium]
	Let $G_{\|v_0}$ be a game. The strategy profile $\bsigma$ is a \emph{Nash equilibrium} --- or \emph{NE} for short --- in $G_{\|v_0}$ if and only if for each player $i$ and for every strategy $\sigma'_i$, called \emph{deviation of $\sigma_i$}, we have the inequality $\mu_i\left(\< \sigma'_i, \bsigma_{-i} \>\right) \leq \mu_i\left(\< \bsigma \>\right)$.
\end{defi}

An SPE is a strategy profile whose all substrategy profiles are NEs.

\begin{defi}[Subgame, substrategy]
	Let $hv$ be a history in the game $G$. The \emph{subgame} of $G$ after $hv$ is the game $\left(\Pi, V, (V_i)_i, E, \mu_{\|hv}\right)_{\|v}$, where $\mu_{\|hv}$ maps each play to its payoff in $G$, assuming that the history $hv$ has already been played: formally, for every $\rho \in \Plays G_{\|hv}$, we have $\mu_{\|hv}(\rho) = \mu(h\rho)$.
    If $\sigma_i$ is a strategy in $G_{\|v_0}$, its \emph{substrategy} after $hv$ is the strategy $\sigma_{i\|hv}$ in $G_{\|hv}$, defined by $\sigma_{i\|hv}(h') = \sigma_i(hh')$ for every $h' \in \Hist_i G_{\|hv}$.
\end{defi}

\begin{rk}
    The initialized game $G_{\|v_0}$ is also the subgame of $G$ after the one-state history $v_0$.
\end{rk}

\begin{defi}[Subgame-perfect equilibrium]
	Let $G_{\|v_0}$ be a game.
	The strategy profile $\bsigma$ is a \emph{subgame-perfect equilibrium} --- or \emph{SPE} for short --- in $G_{\|v_0}$ if and only if for every history $h$ in $G_{\|v_0}$, the strategy profile $\bsigma_{\|h}$ is a Nash equilibrium in the subgame $G_{\|h}$.
\end{defi}

The notion of subgame-perfect equilibrium refines the notion of Nash equilibrium and excludes coercion by non-credible threats.

\begin{exa} Consider the game pictured in Figure~\ref{fig_ne_spe}. It is initialized with initial state $a$, and has two players, player $\Circle$ and player $\Box$, who own respectively the circle and the square vertices. The payoff function assigns to each player a payoff of $1$ for the play $abd^\omega$, and $0$ for all the other plays. Two different strategy profiles are represented here, one by the blue colored transitions, which has outcome $abd^\omega$, one by the red colored ones, which has outcome $acg^\omega$. Both are NEs: clearly, no player can increase their payoff by deviating from the blue choices, and in the case of the red profile, a deviation of player $\Box$ can only lead to the play $acf^\omega$, and a deviation of player $\Circle$ to $abe^\omega$ --- both plays give to both player the payoff $0$.
However, for player $\Box$, going from $b$ to $e$ is not a rational choice, hence the red profile is not an SPE, while the blue profile is one.

\begin{figure} 
	\begin{center}
	\begin{subfigure}[b]{0.45\textwidth}
    	\begin{tikzpicture}[->,>=latex,shorten >=1pt, initial text={}, scale=0.5, every node/.style={scale=0.5}]
    	\node[initial above, state] (a) at (0, 0) {$a$};
    	\node[state, rectangle] (b) at (-2, -2) {$b$};
    	\node[state, rectangle] (c) at (2, -2) {$c$};
    	\node[state] (d) at (-3, -4) {$d$};
    	\node[state] (e) at (-1, -4) {$e$};
    	\node[state] (f) at (1, -4) {$f$};
    	\node[state] (g) at (3, -4) {$g$};
    	\path[->, blue] (a) edge (b);
    	\path[->, red] (a) edge (c);
    	\path[->, red] (b) edge (e);
    	\path[->, blue] (b) edge (d);
    	\path[->, blue] (c) edge (f);
    	\path[->, red] (c) edge (g);
    	\path (d) edge [loop below] node[below] {$\stackrel{\playcircle}{1} \stackrel{\Box}{1}$} (d);
    	\path (e) edge [loop below] node[below] {$\stackrel{\playcircle}{0} \stackrel{\Box}{0}$} (e);
    	\path (f) edge [loop below] node[below] {$\stackrel{\playcircle}{0} \stackrel{\Box}{0}$} (f);
    	\path (g) edge [loop below] node[below] {$\stackrel{\playcircle}{0} \stackrel{\Box}{0}$}(g);
    	\end{tikzpicture}
    	\caption{Two Nash equilibria}
    	\label{fig_ne_spe}
    \end{subfigure}
    \begin{subfigure}[b]{0.45\textwidth}
        	\begin{tikzpicture}[scale=1]
		\draw [->] (0,0) -- (3,0);
		\draw (3,0) node[right] {$x$};
		\draw [->] (0,0) -- (0,3);
		\draw (0,3) node[above] {$y$};
		\fill [gray!40] (0.5, 0.5) -- (1.5, 0.5) -- (0.5, 1.5);
		\fill [blue] (0.5, 1.5) -- (1.5, 2.5) -- (2.5, 1.5) -- (1.5, 0.5);
		\end{tikzpicture}
		\caption{An example for the operator $\dseal$}
		 \label{fig_dseal}
    \end{subfigure}
	
	\end{center}
	\caption{Illustration for preliminary notions}
\end{figure}
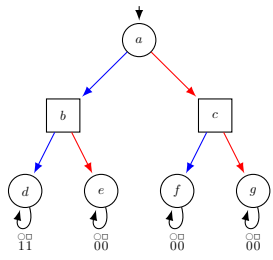
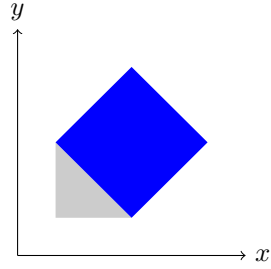
\end{exa}

An $\epsilon$-SPE is a strategy profile which is \emph{almost} an SPE: if a player deviates after some history, they will not be able to improve their payoff by more than a quantity $\epsilon \geq 0$.
Note that a $0$-SPE is an SPE, and conversely.

\begin{defi}[$\epsilon$-SPE]
	Let $G_{\|v_0}$ be a game, and $\epsilon \geq 0$. A strategy profile $\bsigma$ from $v_0$ is an $\epsilon$-SPE if and only if for every history $h$, for every player $i$ and every strategy $\sigma'_i$, we have $\mu_i(\< \bsigma_{-i\|h}, \sigma'_{i\|h} \>) \leq \mu_i(\< \bsigma_{\|h} \>) + \epsilon$.
\end{defi}

\subparagraph{\bf {Mean-payoff games.}}
We now turn to the definition of mean-payoff objectives.
	
	\begin{defi}[Mean-payoff, mean-payoff game]
	    In a graph $(V, E)$, we associate to each mapping $r: E \to \Q$ the \emph{mean-payoff function}:
	    $$\MP_r: h_0 \dots h_n \mapsto \frac{1}{n} \underset{k=0}{\overset{n-1}{\sum}} r_i\left(h_k h_{k+1}\right).$$
		A game $G = \left(\Pi, V, (V_i)_i, E, \mu \right)$ is a \emph{mean-payoff game} if its underlying graph is finite, and if there exists a tuple $(r_i)_{i \in \Pi}$ of reward functions, such that for each player $i$ and every play $\rho$:
		$$\mu_i(\rho) = \liminf_{n \to \infty} \MP_{r_i}(\rho_{\leq n}).$$
	\end{defi}
	
The mapping $r_i$ is called \emph{reward function} of player $i$: it represents the immediate reward that each action grants to player $i$.
The final payoff of player $i$ is their average payoff along the play, classically defined as the limit inferior\footnote{An alternative definition of mean-payoff games exists, with a limit superior instead of inferior. While in zero-sum one dimensional games, the two definitions lead to the same notion of optimality, this is not the case when considering multiple dimensions, see e.g.~\cite{DBLP:journals/iandc/VelnerC0HRR15}.
All the results presented in this paper apply only on mean-payoff games defined with a limit inferior.
}
over $n$ (since the limit may not be defined) of the average payoff after $n$ steps.
When the context is clear, we liberally write $\MP_i(h)$ for $\MP_{r_i}(h)$, and $\MP(h)$ for the tuple $(\MP_i(h))_i$, as well as $r(uv)$ for the tuple $(r_i(uv))_i$.

In the sequel, we develop a worst-case optimal algorithm to solve the \emph{$\epsilon$-SPE threshold problem}, which is a generalization of the \emph{SPE threshold problem}, defined as follows.

\begin{defi}[$\epsilon$-SPE threshold problem]
    Given a rational number $\epsilon \geq 0$, a mean-payoff game $G_{\|v_0}$ and two thresholds $\bx, \by \in \Q^\Pi$, does there exist an $\epsilon$-SPE $\bsigma$ in $G_{\|v_0}$ such that $\bx \leq \mu(\< \bsigma \>) \leq \by$?
\end{defi}

That problem is already known, by \cite{Concur}, to be $2\ExpTime$ and $\NP$-hard.
The proof given in that paper does also show that $\NP$-hardness still holds when $\epsilon$ is fixed to $0$.
Let us also add that the existence of an SPE in a given mean-payoff game, i.e. the same problem with no thresholds and with $\epsilon = 0$, is itself $\NP$-hard.

\begin{defi}[SPE existence problem]
    Given a mean-payoff game $G_{\|v_0}$, does there exist an SPE in $G_{\|v_0}$?
\end{defi}

\begin{lm}[App.~\ref{pf_np_hard}] \label{lm_np_hard}
    The SPE existence problem is $\NP$-hard.
\end{lm}

\subparagraph{{\bf Set of possible payoffs.}}
A first important result that we need is the characterization of the set of possible payoffs in a mean-payoff game, which has been introduced in \cite{DBLP:conf/concur/ChatterjeeDEHR10}.
Given a graph $(V, E)$, we write $\SC(V, E)$ the set of simple cycles it contains.
Given a finite set $D$ of dimensions and a set $X \subseteq \R^D$, we write $\Conv X$ the convex hull of $X$.
We will often use the subscript notation $\Conv_{x \in X} f(x)$ for the set $\Conv f(X)$.

\begin{defi}[Downward sealing] \label{def_dseal}
    Given a set $Y \subseteq \R^D$, the \emph{downward sealing} of $Y$ is the set $\dseal Y = \left\{\left. \left( \min_{\bz \in Z} z_d \right)_{d \in D} ~\right|~ Z \mathrm{~is~a~finite~subset~of~} Y \right\}.$
\end{defi}

\begin{exa}
    In $\R^2$, if $Y$ is the blue area in Figure~\ref{fig_dseal}, then $\dseal Y$ is the union of the blue area and the gray area.

\end{exa}

\begin{lm}[\cite{DBLP:conf/concur/ChatterjeeDEHR10}]\label{lm_dseal}
    Let $G$ be a mean-payoff game, whose underlying graph is strongly connected.
    The set of the payoffs $\mu(\rho)$, where $\rho$ is a play in $G$, is exactly the set:
    $$\dseal\left( \underset{c \in \SC(V, E)}{\Conv} \MP(c) \right).$$
\end{lm}

\subparagraph{\bf {Two-player zero-sum games.}}
We now recall several definitions and two classical results about two-player zero-sum games.

\begin{defi}[Two-player zero-sum game]
	A \emph{two-player zero sum game} is a game $G$ with $\Pi = \{1, 2\}$ and $\mu_2 = -\mu_1$.
\end{defi}
	
\begin{defi}[Borel game]
	A game $G$ is \emph{Borel} if the function $\mu$, from the set $V^\omega$ equipped with the product topology to the Euclidian space $\R^\Pi$, is Borel, i.e. if, for every Borel set $B \subseteq \R^\Pi$, the set $\mu^{-1}(B)$ is Borel.
\end{defi}

\begin{rk}
    Mean-payoff games are Borel (see~\cite{DBLP:journals/tcs/Chatterjee07}).
\end{rk}
	
\begin{lm}[Determinacy of Borel games, \cite{BorelDeterminacy}] \label{lm_borel_determinacy}
	Let $G_{\|v_0}$ be a zero-sum Borel game, with $\Pi = \{1, 2\}$. Then, we have the following equality:
	$$\sup_{\sigma_1} ~ \inf_{\sigma_2} ~ \mu_1(\< \bsigma \>) = \inf_{\sigma_2} ~ \sup_{\sigma_1} ~ \mu_1(\< \bsigma \>).$$
\end{lm}

That quantity is called \emph{value} of $G_{\|v_0}$, denoted by $\val_1(G_{\|v_0})$.

\begin{defi}[Optimal strategy]
    Let $G_{\|v_0}$ be a zero-sum Borel game, with $\Pi = \{1, 2\}$.
    The strategy $\sigma_1$ is \emph{optimal} in $G_{\|v_0}$ if $\inf_{\sigma_2} \mu_1(\< \sigma_1, \sigma_2 \>) = \val_1(G_{\|v_0})$.
\end{defi}

Let us now define memoryless strategies, and a condition under which they can be optimal.

\begin{defi}[Memoryless strategy]
	A strategy $\sigma_i$ in a game $G_{\|v_0}$ is \emph{memoryless} if for every vertex $v \in V_i$ and for all histories $h$ and $h'$, we have $\sigma_i(hv) = \sigma_i(h'v)$.
\end{defi}

We usually write $\sigma_i(\cdot v)$ for the state $\sigma_i(hv)$ for every $h$.
For every game $G_{\|v_0}$, we write  $\ML\left(G_{\|v_0}\right)$ for the set of memoryless strategies in $G_{\|v_0}$.

\begin{defi}[Shuffling]
    Let $\rho, \eta$ and $\theta$ be three plays in a game $G$.
    The play $\theta$ is a \emph{shuffling} of $\rho$ and $\eta$ if there exist two sequences of indices $k_0 < k_1 < \dots$ and $\l_0 < \l_1 < \dots$ such that $\eta_0 = \rho_{k_0} = \eta_{\l_0} = \rho_{k_1} = \eta_{\l_1} = \dots$, and:
    $$\theta = \rho_0 \dots \rho_{k_0-1} \eta_0 \dots \eta_{\l_0-1} \rho_{k_0} \dots \rho_{k_1-1} \eta_{\l_0} \dots \eta_{\l_1-1} \dots.$$
\end{defi}

\begin{defi}[Convexity, concavity]
    A function $f: \Plays G \to \R$ is \emph{convex} if every shuffling $\theta$ of two plays $\rho$ and $\eta$ satisfies $f(\theta) \geq \min\{f(\rho), f(\eta)\}$.
    It is \emph{concave} if $-f$ is convex.
\end{defi}

\begin{rk}
    Mean-payoff functions, defined with a limit inferior, are convex.
\end{rk}

\begin{lm} \label{lm_memoryless}
    In a two-player zero-sum game played on a finite graph, every player whose payoff function is concave has an optimal strategy that is memoryless.
\end{lm}

\begin{proof}
    According to \cite{DBLP:conf/icalp/Kopczynski06}, this result is true for qualitative objectives, i.e. when $\mu$ can only take the values $0$ and $1$.
    It follows that for every $\alpha \in \R$, if a player $i$, whose payoff function is concave, has a strategy that ensures $\mu_i(\rho) \geq \alpha$ (understood as a qualitative objective), then they have a memoryless one.
    Hence the equality:
    $$\val_1(G_{\|v_0}) = \sup_{\sigma_1 \in \ML(G_{\|v_0})} ~ \inf_{\sigma_2} ~ \mu_1(\< \bsigma \>).$$
    Since the underlying graph $(V, E)$ is finite, memoryless strategies exist in finite number, hence the supremum above is realized by a memoryless strategy $\sigma_1$ that is, therefore, optimal.
\end{proof}

	\section{Requirements and negotiation} \label{sec_negotiation}

%In this section, we recall notions and results from \cite{Concur}, which will be the starting point of our new algorithm.

We now recall some notions and results from \cite{Concur}, which are the starting point of our algorithm.

\subparagraph{\bf {Requirements.}}
In the sequel, we write $\bR$ the set $\R \cup \{\pm \infty\}$.

\begin{defi}[Requirement]
	A \emph{requirement} on the game $G$ is a mapping $\lambda: V \to \bR$.
\end{defi}

For a given state $v$, the quantity $\lambda(v)$ represents the minimal payoff that the player controlling $v$ will require in a play traversing the state $v$.

\begin{defi}[$\lambda$-consistency]
	Let $\lambda$ be a requirement on a game $G$. A play $\rho$ in $G$ is \emph{$\lambda$-consistent} if and only if, for all $i \in \Pi$ and $n \in \N$ with $\rho_n \in V_i$, we have $\mu_i(\rho_{\geq n})~\geq~\lambda(\rho_n)$.
	The set of $\lambda$-consistent plays from a state $v$ is denoted by $\lCons(v)$.
\end{defi}

\begin{rk}
    The set $\lCons(v)$ can be empty, and is not regular in general.
\end{rk}

\begin{defi}[$\lambda$-rationality]
	Let $\lambda$ be a requirement on a mean-payoff game $G$. Let $i \in \Pi$. A strategy profile $\bsigma_{-i}$ is \emph{$\lambda$-rational} if and only if there exists a strategy $\sigma_i$ such that, for every history $hv$ compatible with $\bsigma_{-i}$, the play $\< \bsigma_{\|hv} \>$ is $\lambda$-consistent.
	We then say that the strategy profile $\bsigma_{-i}$ is $\lambda$-rational \emph{assuming} $\sigma_i$.
	The set of $\lambda$-rational strategy profiles in $G_{\|v}$ is denoted by $\lRat(v)$.	
\end{defi}

\subparagraph{\bf {Negotiation.}} \label{ss_def_nego}
In mean-payoff games, as well as in a wider class of games (see \cite{Concur} and \cite{DBLP:journals/mor/FleschP17}), SPEs are characterized by the fixed points of the \emph{negotiation function}, a function from the set of requirements into itself.
We always use the convention $\inf \emptyset = +\infty$.

\begin{defi}[Negotiation function]
	Let $G$ be a game.
	The \emph{negotiation function} is the function that transforms each requirement $\lambda$ on $G$ into a requirement $\nego(\lambda)$ on $G$ defined, for each $i \in \Pi$ and $v \in V_i$, by:
	$$\nego(\lambda)(v) = \inf_{\bsigma_{-i} \in \lRat(v)} \sup_{\sigma_i} \mu_i(\< \bsigma\>).$$
\end{defi}

The quantity $\nego(\lambda)(v)$ is the best payoff the player controlling the state $v$ can enforce if the other players play rationally with regards to the requirement $\lambda$.

\begin{rk}
    The negotiation function satisfies the following properties.
    
    \begin{itemize}
        \item It is monotone: if $\lambda \leq \lambda'$ (for the pointwise order), then $\nego(\lambda) \leq \nego(\lambda')$.
        
        \item It is also non-decreasing: for every $\lambda$, we have $\lambda \leq \nego(\lambda)$.
        
        \item There exists a $\lambda$-rational strategy profile from $v$ against the player controlling $v$ if and only if $\nego(\lambda)(v) \neq +\infty$.
    \end{itemize}
\end{rk}

 \subparagraph{\bf {Link with SPEs.}}
The SPE outcomes in a mean-payoff game are characterized by the fixed points of the negotiation function, or equivalently by its least fixed point.
That result can be extended to $\epsilon$-SPEs.
To that end, we recall the notion of $\epsilon$-fixed points of a function.

\begin{defi}[$\epsilon$-fixed point]
	Let $\epsilon \geq 0$, let $D$ be a finite set and let $f: \bR^D \to \bR^D$ be a mapping. A tuple $\bx \in \R^D$ is a \emph{$\epsilon$-fixed point} of $f$ if for each $d \in D$, for $\by = f(\bx)$, we have $y_d \in [x_d - \epsilon, x_d + \epsilon]$.
\end{defi}

\begin{rk}
    A $0$-fixed point is a fixed point, and conversely.
\end{rk}

\begin{lm}[\cite{Concur}] \label{lm_spe}
	Let $G_{\|v_0}$ be a mean-payoff game, and let $\epsilon \geq 0$.
	The negotiation function on $G$ has a least $\epsilon$-fixed point $\lambda^*$, and given a play $\rho$ in $G_{\|v_0}$, the three following assertions are equivalent: (1) the play $\rho$ is an $\epsilon$-SPE outcome; (2) the play $\rho$ is $\lambda$-consistent for some $\epsilon$-fixed point $\lambda$ of the negotiation function; (3) the play $\rho$ is $\lambda^*$-consistent.
\end{lm}

\subparagraph{\bf {The abstract negotiation game.}}
Given $\lambda$ and $u \in V$, the quantity $\nego(\lambda)(u)$ can be characterized as the value of a \emph{negotiation game}, a two-player zero-sum game opposing the player \emph{Prover}, who simulates a $\lambda$-rational strategy profile and wants to minimize player $i$'s payoff, and the player \emph{Challenger}, who simulates player $i$'s reaction by accepting or refusing Prover's proposals.
Two negotiation games were defined in \cite{Concur}.
Conceptually simpler, the \emph{abstract negotiation game} $\Abs_{\lambda i}(G)_{\|u}$ unfolds as follows:

\begin{itemize}
    \item from the state $v$, Prover chooses a $\lambda$-consistent play $\rho$ from the state $v$ and proposes it to Challenger.
    If Prover has no play to propose, the game is over and Challenger gets the payoff $+\infty$.
        
    \item Once a play $\rho$ has been proposed, Challenger can accept it.
    Or he can \emph{deviate}, and choose a prefix $\rho_{\leq k}$ with $\rho_k \in V_i$ and a new transition $\rho_k w \in E$.
        
    \item In the former case, the game is over.
    In the latter, it starts again from the state $w$.
\end{itemize}

If Challenger finally accepts a proposal $\rho$, then his payoff is $\mu_i(\rho)$.
If he deviates infinitely often, then Prover's proposals and his deviations construct a play $\dpi = \rho^{(0)}_{\leq k_0} \rho^{(1)}_{\leq k_1} \rho^{(2)}_{\leq k_2} \dots$.
Then, Challenger's payoff is $\mu_i(\dpi)$.
It has been proved in \cite{Concur} that the equality $\nego(\lambda)(u) = \val_\C(\Abs_{\lambda i}(G)_{\|u})$ holds.
Thus, the abstract negotiation game captures a first intuition on how the negotiation function can be computed.

\begin{exa}
    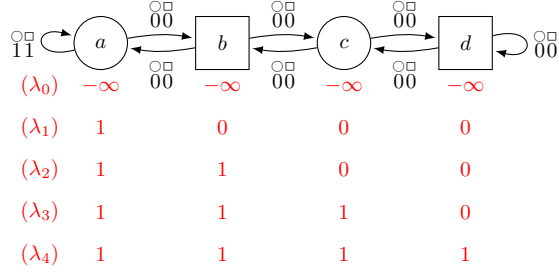
\begin{figure}
    \centering
		\begin{tikzpicture}[->,>=latex,shorten >=1pt, initial text={}, scale=0.8, every node/.style={scale=0.8}]
		\node[state] (a) at (0, 0) {$a$};
		\node[state, rectangle] (b) at (2, 0) {$b$};
		\node[state] (c) at (4, 0) {$c$};
		\node[state, rectangle] (d) at (6, 0) {$d$};
		\path (a) edge [loop left] node {$\stackrel{\playcircle}{1} \stackrel{\Box}{1}$} (a);
		\path[bend left=10] (a) edge node[above] {$\stackrel{\playcircle}{0} \stackrel{\Box}{0}$} (b);
		\path[bend left=10] (b) edge node[below] {$\stackrel{\playcircle}{0} \stackrel{\Box}{0}$} (a);
		\path[bend left=10] (b) edge node[above] {$\stackrel{\playcircle}{0} \stackrel{\Box}{0}$} (c);
		\path[bend left=10] (c) edge node[below] {$\stackrel{\playcircle}{0} \stackrel{\Box}{0}$} (b);
		\path[bend left=10] (c) edge node[above] {$\stackrel{\playcircle}{0} \stackrel{\Box}{0}$} (d);
		\path[bend left=10] (d) edge node[below] {$\stackrel{\playcircle}{0} \stackrel{\Box}{0}$} (c);
		\path (d) edge [loop right] node {$\stackrel{\playcircle}{0} \stackrel{\Box}{0}$} (d);

		\node[red] (l0) at (-1, -0.7) {$(\lambda_0)$};
		\node[red] (l0a) at (0, -0.7) {$-\infty$};
		\node[red] (l0b) at (2, -0.7) {$-\infty$};
		\node[red] (l0c) at (4, -0.7) {$-\infty$};
		\node[red] (l0d) at (6, -0.7) {$-\infty$};
		
		\node[red] (l1) at (-1, -1.4) {$(\lambda_1)$};
		\node[red] (l1a) at (0, -1.4) {$1$};
		\node[red] (l1b) at (2, -1.4) {$0$};
		\node[red] (l1c) at (4, -1.4) {$0$};
		\node[red] (l1d) at (6, -1.4) {$0$};
		
		\node[red] (l2) at (-1, -2.1) {$(\lambda_2)$};
		\node[red] (l2a) at (0, -2.1) {$1$};
		\node[red] (l2b) at (2, -2.1) {$1$};
		\node[red] (l2c) at (4, -2.1) {$0$};
		\node[red] (l2d) at (6, -2.1) {$0$};
		
		\node[red] (l3) at (-1, -2.8) {$(\lambda_3)$};
		\node[red] (l3a) at (0, -2.8) {$1$};
		\node[red] (l3b) at (2, -2.8) {$1$};
		\node[red] (l3c) at (4, -2.8) {$1$};
		\node[red] (l3d) at (6, -2.8) {$0$};
		
		\node[red] (l4) at (-1, -3.5) {$(\lambda_4)$};
		\node[red] (l4a) at (0, -3.5) {$1$};
		\node[red] (l4b) at (2, -3.5) {$1$};
		\node[red] (l4c) at (4, -3.5) {$1$};
		\node[red] (l4d) at (6, -3.5) {$1$};
		\end{tikzpicture}
    \caption{Iterations of the negotiation function}
    \label{fig_ex_nego}
\end{figure}

    Let $G$ be the game of Figure~\ref{fig_ex_nego}, where each edge is labelled by the rewards $r_\playcircle$ and $r_\Box$.
    Below the states, we present the requirements $\lambda_0: v \mapsto -\infty$, $\lambda_1 = \nego(\lambda_0)$, $\lambda_2 = \nego(\lambda_1)$, $\lambda_3 = \nego(\lambda_2)$, and $\lambda_4 = \nego(\lambda_3)$.
    Let us explicate those computations, using the abstract negotiation game.
	From $\lambda_0$ to $\lambda_1$: since every play is $\lambda_0$-consistent, Prover can always propose whatever she wants.
	From the state $a$, whatever she (trying to minimize player $\Circle$'s payoff) proposes, Challenger can always make player $\Circle$ deviate in order to loop on the state $a$.
	Then, in the game $G$, player $\Circle$ gets the payoff $1$, hence $\lambda_1(a) = 1$.
	From the state $b$, Prover (trying to minimize player $\Box$'s payoff) can propose the play $(bc)^\omega$.
	If Challenger makes player $\Box$ deviate to go to the state $a$, then Prover can propose the play $a(bc)^\omega$.
	Even if Challenger makes player $\Box$ deviate infinitely often, he cannot give him more than the payoff $0$, hence $\lambda_1(b) = 0$.
	Similar situations happen from the states $c$ and $d$, hence $\lambda_1(c) = \lambda_1(d) = 0$.
	From $\lambda_1$ to $\lambda_2$: now, from the state $b$, whatever Prover proposes at first, Challenger can make player $\Box$ deviate and go to the state $a$.
	From there, since we have $\lambda_1(a) = 1$, Prover has to propose a play in which player $\Circle$ gets the payoff $1$.
	The only such plays do also give the payoff $1$ to player $\Box$, hence $\lambda_2(b) = 1$.
	Similar situations explain $\lambda_3(c) = 1$, and $\lambda_4(c) = 1$.
	Finally, plays ending with the loop $a^\omega$ are all $\lambda_4$-consistent, hence Prover can always propose them, hence the requirement $\lambda_4$ is a fixed point of the negotiation function --- and therefore the least.
	By Lemma~\ref{lm_spe}, the SPE plays in $G$ are exactly the plays in which both player $\Circle$ and player $\Box$ get the payoff $1$.
\end{exa}

\subparagraph{\bf {The concrete negotiation game.}}
The abstract negotiation game cannot be directly used for an algorithmic purpose, since it has an infinite state space.
However, it can be turned into a game on a finite graph if Prover does not propose plays as a whole, but edge by edge.
In the \emph{concrete negotiation game} $\Conc_{\lambda i}(G)_{\|(u, \{u\})}$, the states controlled by Prover have the form $(v, M)$, where $M \subseteq V$ memorizes the states seen since the last time Challenger did deviate, in order to control that the play Prover is constructing since that moment is effectively $\lambda$-consistent: for each $u \in M$, Prover has to give to the player controlling $u$ at least the payoff $\lambda(u)$.
Similarly, the states controlled by Challenger are of the form $(vv', M)$, where $vv' \in E$ is an edge proposed by Prover.
The game unfolds as follows:

\begin{itemize}
    \item from the state $(v, M)$, Prover chooses a transition $vv'$ and proposes it to Challenger.
        
    \item Once a transition $vv'$ has been proposed, Challenger can accept it.
    Or, if $v \in V_i$, he can \emph{deviate}, and choose a new transition $vw$.
        
    \item If the former case, the game starts again from the state $(v', M \cup \{v'\})$.
    In the latter, it starts from the state $(w, \{w\})$.
\end{itemize}
    
For every play $\pi = (\rho_0, M_0) (\rho_0\rho'_0, M_0) (\rho_1, M_1) (\rho_1\rho'_1, M_1) \dots$ in $\Conc_{\lambda i}(G)$, we write $\dpi = \rho_0 \rho_1 \dots$ the play in $G$ constructed by Prover's proposals and Challenger's deviations.
Then, Challenger's payoff in the play $\pi$ is either $+\infty$ if there exists an index $k$ such that the suffix $\pi_{\geq 2k}$ contains no deviation and $\dpi_{\geq k}$ is not $\lambda$-consistent, and $\mu_i(\dpi)$ otherwise.

In \cite{Concur}, a first algorithm was proposed to solve the $\epsilon$-SPE threshold problem, using the fact that in the concrete negotiation game, Challenger has a memoryless optimal strategy, to design a complete representation of the negotiation function, and to compute its least $\epsilon$-fixed point.
However, that algorithm requires doubly exponential time, because it needs to enumerate all the memoryless strategies available for Challenger, whose number is exponential in the size of the concrete game, itself exponential in the size of $G$.
Here, we make use of the concrete negotiation game only to bound the size of the least $\epsilon$-fixed point: our algorithm will use a third negotiation game, the \emph{reduced negotiation game}.

    \section{Size of the least $\epsilon$-fixed point}

In this section, after having recalled some results about the sizes of solutions to linear equations and inequations, we prove that the least $\epsilon$-fixed point of the negotiation function in a game $G$ has a size that is polynomial in the size of $G$ and $\epsilon$.
The first piece of the witnesses identifying positive instances of the $\epsilon$-SPE threshold problem will then be an $\epsilon$-fixed point of the negotiation function of polynomial size.

\subparagraph{\bf {About size, equations and inequations.}}
We define here the notion of size that we use.

\begin{defi}[Size]
    The \emph{size} of a rational number $r = \frac{p}{q}$, where $p, q \in \Z$ are co-prime, is the quantity $\Vert r\Vert = 1 + \lceil \log_2(|p|+1) \rceil + \lceil \log_2(|q|+1) \rceil$.
    The size of an irrational number is $+\infty$.
    The size of the infinite numbers is $\Vert +\infty \Vert = \Vert -\infty \Vert = 1$.
    The size of a tuple $\bx \in O^D$, where $D$ is a set and $O$ is a set of objects for which the notion of size has been defined, is the quantity $\card D + \sum_{d \in D} \Vert x_d \Vert$.
    Similarly, the size of a function $f: D \to X$ is the quantity $\card D + \sum_{d \in D} \Vert f(d) \Vert$, and the size of a set $X \subseteq O$ is the quantity $\card X + \sum_{x \in X} \Vert x \Vert$.
\end{defi}

The proof of Theorem~\ref{thm_size_lambda} below requires the manipulation of \emph{polytopes}, e.g. downward sealings of convex hulls (from Lemma~\ref{lm_dseal}), expressed as solution sets of \emph{systems of linear inequations}.

\begin{defi}[Linear equations, inequations, systems]
    Let $D$ be a finite set.
    A \emph{linear equation} in $\R^D$ is a pair $(\ba, b) \in \left(\R^D \setminus \left\{\bzero\right\}\right) \times \R$.
    The \emph{solution set} of the equation $(\ba, b)$ is the set $\Sol_=(\ba, b) = \{ \bx \in \R^D ~|~ \ba \cdot \bx = b\}$, where $\cdot$ denotes the canonical scalar product on the euclidian space $\R^D$.
    A set $X \subseteq \R^D$ is a \emph{hyperplane} of $\R^D$ if it is the solution set of some linear equation.
    A \emph{system of linear equations} is a finite set $\Sigma$ of linear equations.
    The \emph{solution set} of the system $\Sigma$ is the set $\Sol_= \Sigma = \bigcap_{(\ba, b) \in \Sigma} \Sol_=(\ba, b)$.
    A set $X \subseteq \R^D$ is a \emph{linear subspace} of $\R^D$ if it is the solution set of some system of linear equations.
    
    A \emph{linear inequation} in $\R^D$ is a pair $(\ba, b) \in \left(\R^D \setminus \left\{\bzero\right\}\right) \times \R$.
    The \emph{solution set} of the inequation $(\ba, b)$ is the set $\Sol_\geq(\ba, b) = \{ \bx \in \R^D ~|~ \ba \cdot \bx \geq b\}$.
    A set $X \subseteq \R^D$ is a \emph{half-space} of $\R^D$ if it is the solution set of some linear inequation.
    A \emph{system of linear inequations} is a finite set $\Sigma$ of linear inequations.
    The \emph{solution set} of the system $\Sigma$ is the set $\Sol_\geq \Sigma = \bigcap_{(\ba, b) \in \Sigma} \Sol_\geq(\ba, b)$.
    A set $X \subseteq \R^D$ is a \emph{polyhedron} of $\R^D$ if it is the solution set of some system of linear inequations $\Sigma$.
    A \emph{vertex} of $X$ is a point $\bx \in \R^D$ such that $\{\bx\} = \Sol_=(\Sigma')$ for some subset $\Sigma' \subseteq \Sigma$.
    A \emph{polytope} is a bounded polyhedron.
\end{defi}

\begin{rk}
    Polyhedra are closed sets.
    The polytopes of $\R^D$ are exactly the sets of the form $\Conv(S)$, where $S$ is a finite subset of $\R^D$.
\end{rk}

\begin{lm}[\cite{DBLP:conf/concur/ChatterjeeDEHR10}] \label{lm_dseal_size}
    Let $\Sigma$ be a system of inequations, and let $X = \Sol_{\geq}(\Sigma)$.
    The set $\dseal X$ is itself a polyhedron, and there exists a system of inequations $\Sigma'$ such that $\dseal X = \Sol_\geq (\Sigma')$ and that for every $(\ba', b') \in \Sigma'$, there exists $(\ba, b) \in \Sigma$ with $\lv (\ba', b') \rv \leq \lv (\ba, b) \rv$.
\end{lm}

\begin{lm}[\cite{DBLP:conf/cav/BrenguierR15}, Theorem~1] \label{lm_system_eq_to_point}
    There exists a polynomial $P_1$ such that, for every system of equations $\Sigma$, there exists a point $\bx \in \Sol_= \Sigma$, such that $\lv \bx \rv \leq P_1\left(\max_{(\ba, b) \in \Sigma} \lv (\ba, b) \rv\right)$.
\end{lm}

\begin{cor} \label{cor_size_vertices}
    For every system of inequations $\Sigma$, each vertex $\bx$ of the polyhedron $\Sol_{\geq}(\Sigma)$ has size $\lv x \rv \leq P_1\left(\max_{(\ba, b) \in \Sigma} \lv (\ba, b) \rv\right)$.
\end{cor}

Note that in Lemma~\ref{lm_dseal_size}, in Lemma~\ref{lm_system_eq_to_point} and in Corollary~\ref{cor_size_vertices}, the number of equations or inequations has no influence.
A consequence of Lemma~\ref{lm_system_eq_to_point} is the following result.

\begin{lm}[App.~\ref{pf_conv_to_system}] \label{lm_conv_to_system}
    There exists a polynomial $P_2$ such that, for each finite set $D$ and every finite subset $X \subseteq \R^D$, there exists a system of linear inequations $\Sigma$, such that $\Sol_\geq(\Sigma) = \Conv(X)$ and $\lv (\ba, b) \rv \leq P_2(\lv X \rv)$ for every $(\ba, b) \in \Sigma$.
\end{lm}

\subparagraph{\bf {Size of the least $\epsilon$-fixed point.}}
The following theorem bounds the size of the least fixed point of the negotiation function.
As we used the notation $\bx$ for tuples so far, we use the notation $\bbx$ for tuples of tuples.

\begin{thm}[App.~\ref{pf_size_lambda}] \label{thm_size_lambda}
    There exists a polynomial $P_3$ such that for every mean-payoff game $G$, the least $\epsilon$-fixed point $\lambda^*$ of the negotiation function has size $\lv \lambda^* \rv \leq P_3(\lv G \rv + \lv \epsilon \rv)$.
\end{thm}

\begin{proof}[Proof sketch]
    It has been proved in \cite{Concur} that Challenger has a memoryless optimal strategy in every concrete negotiation game.
    Given a requirement $\lambda$, a player $i$, a state $v \in V_i$ and a memoryless strategy $\tau_\C$, we can construct the set of payoff vectors $\mu(\dpi)$, where $\pi$ is a play in $\Conc_{\lambda i}(G)_{\|(v, \{v\})}$ compatible with $\tau_\C$, as a union of polytopes defined using Lemma~\ref{lm_dseal}.
    If we intersect the upward closures of those sets, then $\nego(\lambda)(v)$ is equal to the least value $x_i$, where $\bx$ belongs to that intersection.
    Therefore, if $X_\lambda \subseteq \R^{V \times \Pi}$ is the product of those intersections, then for each $i$ and $v \in V_i$, we have $\nego(\lambda) = \inf\{x_{vi} ~|~ \bbx \in X_\lambda\}$.

    To each tuple of tuples $\bbx \in \R^{V \times \Pi}$, we associate the requirement $\lambda_{\bbx}$ defined by $\lambda_{\bbx}(v) = x_{vi} - \epsilon$ for each $i \in \Pi$ and $v \in V_i$.
    Then, we define $X = \left\{\bbx ~\left|~ \bbx \in X_{\lambda_{\bbx}} \right.\right\}$, and we show that any requirement $\lambda$ is an $\epsilon$-fixed point of the negotiation function if and only $\lambda = \lambda_{\bbx}$ for some $\bbx \in X$.
    Then, the set $X$ is itself a union of polyhedra, hence the linear mapping $\bbx \mapsto \sum_v \lambda_{\bbx}(v)$ has its minimum over $X$ on some vertex $\bbx$ of one of those polyhedra.
    The requirement $\lambda^*$ is equal to $\lambda_{\bbx}$, hence its size can be bounded using Corollary~\ref{cor_size_vertices}.
\end{proof}

    \section{Constrained existence of a $\lambda$-consistent play}

We claim that a non-deterministic algorithm can recognize the positive instances of the $\epsilon$-SPE threshold problem by guessing an $\epsilon$-fixed point $\lambda$ of the negotiation function.
Once $\lambda$ has been guessed, according to Lemma~\ref{lm_spe}, two assertions must be proved: on the one hand, that there exists a $\lambda$-consistent play between the two desired thresholds, and on the other hand, that $\lambda$ is actually an $\epsilon$-fixed point of the negotiation function.
The latter will be handled later through the concept of \emph{reduced negotiation game}.
Now, we tackle the former, and provide the second piece of our notion of witness: to prove the existence of a $\lambda$-consistent play $\rho$ with $\bx \leq \mu(\rho) \leq \by$, we need to guess the sets $W = \Inf(\rho)$ and $W' = \Occ(\rho)$, and a tuple of tuples $\bbalpha \in [0, 1]^{\Pi \times \SC(W)}$ indicating how $\rho$ combines the cycles of $W$, i.e. such that:
$$\mu(\rho) = \left( \min_{j \in \Pi} \sum_{c \in \SC(W)} \alpha_{jc} \MP_i(c) \right)_i.$$

\begin{thm} \label{thm_constrained_existence_lambda_cons}
    There exists a polynomial $P_4$ such that for every mean-payoff game $G_{\|v_0}$, for every $\bx, \by \in \R^V$, and for every requirement $\lambda$ on $G$, there exists a $\lambda$-consistent play $\rho$ in $G_{\|v_0}$ satisfying $\bx \leq \mu(\rho) \leq \by$ if and only if there exist two sets $W \subseteq W' \subseteq V$ and a tuple of tuples $\bbalpha \in [0, 1]^{\Pi \times \SC(W)}$ such that:
    \begin{itemize}
        \item the set $W$ is strongly connected in $(V, E)$, and accessible from the state $v_0$ using only and all the states of $W'$;
        
        \item for each player $i$, we have $\sum_c \alpha_{ic} = 1$, and:
        $$x_i \leq \min_{j \in \Pi} \sum_{c \in \SC(W)} \alpha_{jc} \MP_i(c) \leq y_i;$$
        
        \item for each player $i$ and $v \in W \cap V_i$, we have:
        $$\min_{j \in \Pi} \sum_{c \in \SC(W)} \alpha_{jc} \MP_i(c) \geq \lambda(v);$$
        
        \item $\lv \bbalpha \rv \leq P_4(\lv G, \bx, \by, \lambda \rv)$.
    \end{itemize}
\end{thm}

\begin{proof}
    Let us first notice that given a set $X \subseteq \R^\Pi$, the elements of the set $\dseal (\Conv X)$ are exactly the tuples of the form:
    $$\left( \min_{j \in \Pi} \sum_{x \in X} \alpha_{jx} x \right)_{i \in \Pi}$$
    for some tuple $\bbalpha \in \R^{\Pi \times X}$ satisfying $\sum_x \alpha_{ix} = 1$ for each $x$.

    Now, let us assume that $W$, $W'$ and $\bbalpha$ exist.
    Then, there exists a play $\eta$ with $\Occ(\eta) = \Inf(\eta) = W$ with payoff vector:
    $$\mu(\eta) = \left( \min_{j \in \Pi} \sum_{x \in X} \alpha_{jx} x \right)_{i \in \Pi}.$$
    Moreover, since $W$ is accessible from $v_0$ using all and only the vertices of $W'$, there exists a history $h \eta_0$ from $v_0$ to $\eta_0$ with $\Occ(h) = W'$.
    Then, the play $\rho = h \eta$ is $\lambda$-consistent and satisfies $\bx \leq \mu(\rho) \leq \by$.
    
    Conversely, if the play $\rho$ exists: let $W = \Inf(\rho)$ and $W' = \Occ(\rho)$.
    The polytope:
    $$Z = \left\{\mu(\eta) ~\left|~ \begin{matrix}
        \eta \in \lCons(G_{\|v_0}), \\
        \Inf(\eta) = W, \\
        \Occ(\eta) = W', \\
        \mathrm{and~} \bx \leq \mu(\eta) \leq \by
    \end{matrix} \right. \right\}
    = \left\{ \bz \in \dseal \left( \underset{c \in \SC(W)}{\Conv} \MP(c) \right) ~\left|~ \begin{matrix}
        \bx \leq \bz \leq \by, \mathrm{~and} \\
        \forall i, \forall v \in W' \cap V_i, z_i \geq \lambda(v)
    \end{matrix} \right. \right\}$$
    (the equality holds by Lemma~\ref{lm_dseal}) is nonempty (it contains at least $\mu(\rho)$).
    By Lemma~\ref{lm_conv_to_system}, the set $\Conv_{c \in SC(W)} \MP(c)$ is defined by a system of inequations which all have size $\lv (\ba, b) \rv \leq P_2\left(\max_c \lv \MP(c) \rv\right)$.
    Since by Lemma~\ref{lm_dseal_size}, the inequations defining $\dseal \left( \Conv_{c \in SC(W)} \MP(c)\right)$ are not larger, there exists a polynomial $P_6$, independent of $G, \bx, \by$ and $\lambda$, such that $Z$ is defined by a system of inequations $\Sigma$ such that for every $(\ba, b) \in \Sigma$, we have $\lv (\ba, b) \rv \leq P_6(\lv (G, \bx, \by, \lambda) \rv)$. 
    Therefore, by Corollary~\ref{cor_size_vertices}, the polytope $Z$ admits a vertex $\bz$ of size $\lv \bz \rv \leq P_1(P_6(\lv (G, \bx, \by, \lambda) \rv))$.
    
    Then, since we have $\bz \in \dseal \left( \underset{c \in \SC(W)}{\Conv} \MP(c) \right)$, that vertex is, according to Definition~\ref{def_dseal}, of the form:
    $$\bz = \left( \min_j \sum_c \alpha_{jc} \MP_j(c) \right)_i$$
    for some tuple of tuples $\bbalpha \in [0, 1]^{\Pi \times \SC(W)}$ with $\sum_c \alpha_{ic} = 1$ and having, by Corollary~\ref{cor_size_vertices} again, size $\lv \bbalpha \rv \leq P_1\left( \max_{i \in \Pi} \sum_{c \in \SC(W)} \lv \MP_i(c) \rv + \lv z_i \rv\right)$, i.e. $\lv \bbalpha \rv \leq P_4(\lv (G, \bx, \by, \lambda) \rv)$ for some polynomial $P_4$ independent of $G, \bx, \by$ and $\lambda$.
\end{proof}

Now, we need the third piece of our witness, which will be evidence of the fact that the requirement $\lambda$ is an $\epsilon$-fixed point of the negotiation function.

    \section{The reduced negotiation game} \label{sec_reduced_game}

The abstract negotiation game has an infinite (and uncountable) state space in general, and the concrete negotiation game has an exponential one.
In \cite{CSL}, the infiniteness of the abstract negotiation game has been handled in the case of parity games, by proving that Prover has an optimal strategy that is memoryless, and that proposes only simple plays with a finite representation.
Unfortunately, this result does not apply to mean-payoff games, where Prover needs infinite memory in general.
That fact is illustrated in the next example.

\begin{exa}
    In the game of Figure~\ref{fig_inf_memory}, the requirement $\lambda$ defined by $\lambda(a) = \lambda(b) = 1$ is a fixed point of the negotiation function (it is actually the least fixed point).
    Indeed, from the state $a$ (the situation is symmetrical from the state $b$), consider the strategy for Prover that proposes always, from the state $v$, the play $v b^{|h|^2} (a^3b^3)^\omega$, where $h$ is the history that has already been constructed by her proposals and Challenger's deviations.
    If Challenger accepts such a play, then he gets the payoff $1$.
    If he deviates infinitely often, then Prover loops longer and longer on the state $b$, and he also gets the payoff $1$.
    The loop on $b$ corresponds to what we will call later a \emph{punishing cycle}.
    Now, if Prover uses only finite memory, Challenger can get a payoff better than $1$ by always deviating and go to $b$ as soon as he can: then, edges giving to player $\Circle$ the reward $2$ will occur with a nonzero frequency.
\end{exa}

However, the plays proposed by Prover in the previous example are very similar: only the number of repetitions of the loop $b$ does increase.
More generally, one observes that Prover can play optimally while always proposing a play of the form $h c^n \rho$, where $h$, $c$ and $\rho$ are constant, and only the number $n$ increases, quadratically with the time --- so that Challenger's payoff is dominated by the mean-payoff $\MP_i(c)$ if he deviates infinitely often.

\begin{defi}[Punishment family]
    A \emph{punishment family} is a set of plays of the form:
    $$\left\{\left. h c^n \rho ~\right|~ n > 0, \mu(\rho) = \bx, \Occ(\rho) = W \right\}$$
    where $h$ is a simple history, $c$ is a nonempty simple cycle, and where $W \subseteq V$ and $\bx \in \R^\Pi$.
    The cycle $c$ is called its \emph{punishing cycle}.
    For every $\beta \in \N$, a \emph{$\beta$-punishment family} is a punishment family with $\lv \bx \rv \leq \beta$.
    A $\beta$-punishment family is represented by the data $h$, $c$, $\mu(\rho)$ and $\Occ(\rho)$, and that representation has a size smaller than or equal to the quantity $3 \card V \lceil \log_2(\card V + 1) \rceil + \beta$.
\end{defi}

We write $h c^\infty \rho$ for the punishment family $\{ h c^n \rho' ~|~ n > 0, \mu(\rho') = \mu(\rho), \Occ(\rho') = \Occ(\rho)\}$.
Beware that the play $\rho$ matters only for its payoff vector and the vertices it traverses: if $\Occ(\rho) = \Occ(\rho')$ and $\mu(\rho) = \mu(\rho')$, then $h c^\infty \rho = h c^\infty \rho'$.
We write $\mu(h c^\infty \rho)$ for the common payoff vector of all elements of $h c^\infty \rho$, and we will say that $h c^\infty \rho$ is $\lambda$-consistent if all its elements are (or equivalently, if one of its elements is).
Let us clarify that a punishment family is not an equivalence class: for example, in the game of Figure~\ref{fig_inf_memory}, the play $ab^\omega$ belongs to both $a^\infty b^\omega$ and $ab^\infty b^\omega$, which are distinct.
We can now define the \emph{reduced negotiation game}, where Prover proposes $\beta$-punishment families instead of plays.

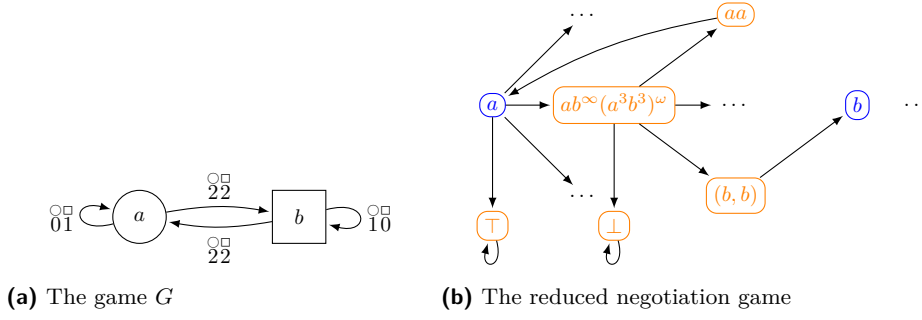
\begin{figure}
    \centering
        \begin{subfigure}[b]{0.4\textwidth}
        \centering
		\begin{tikzpicture}[->,>=latex,shorten >=1pt, initial text={}, scale=0.7, every node/.style={scale=0.8}]
		\node[state] (a) at (0, 0) {$a$};
		\node[state, rectangle] (b) at (3, 0) {$b$};
		\path[bend left = 10] (a) edge node[above] {$\stackrel{\playcircle}{2}\stackrel{\Box}{2}$} (b);
		\path[bend left = 10] (b) edge node[below] {$\stackrel{\playcircle}{2}\stackrel{\Box}{2}$} (a);
		\path (a) edge [loop left] node {$\stackrel{\playcircle}{0}\stackrel{\Box}{1}$} (a);
		\path (b) edge [loop right] node {$\stackrel{\playcircle}{1}\stackrel{\Box}{0}$} (b);
		\end{tikzpicture}
		\caption{The game $G$}
		\label{fig_inf_memory}
	\end{subfigure}
	\begin{subfigure}[b]{0.5\textwidth}
	\centering
	\begin{tikzpicture}[->, >=latex,shorten >=1pt, initial text={}, scale=0.8, every node/.style={scale=0.8}, rounded corners]
	
	\node[draw, blue] (a) at (0, 0) {$a$};
	\node[draw, orange] (red) at (2, 0) {$a b^\infty (a^3b^3)^\omega$};
	\node[draw, orange] (ac) at (4, 1.5) {$aa$};
	\node[draw, orange] (abc) at (4, -1.5) {$(b, b)$};
	\node[draw, blue] (c) at (6, 0) {$b$};
	\node (dots1) at (1.5, 1.5) {\dots};
	\node (dots2) at (1.5, -1.5) {\dots};
	\node (dots3) at (4, 0) {\dots};
	\node (dots4) at (7, 0) {\dots};
	\node[draw, orange] (bot) at (0, -2) {$\top$};
	\node[draw, orange] (top) at (2, -2) {$\bot$};
	
	\path (a) edge (red);
	\path (red) edge (ac);
	\path (red) edge (abc);
	\path[bend right=10] (ac) edge (a);
	\path (abc) edge (c);
	\path (a) edge (dots1);
	\path (a) edge (dots2);
	\path (red) edge (dots3);
	\path (a) edge (bot);
	\path (bot) edge[loop below] (bot);
	\path (red) edge (top);
	\path (top) edge[loop below] (top);
	\end{tikzpicture}
\caption{The reduced negotiation game} \label{fig_reduced}
\end{subfigure}
    \caption{A game on which Prover needs infinite memory}
\end{figure}

\begin{defi}[Reduced negotiation game] \label{def_reduced_game}
    Let $G$ be a mean-payoff game, let $\lambda$ be a requirement, let $i$ be a player, let $v_0 \in V_i$ and let $\beta$ be a natural integer.
    The corresponding \emph{reduced negotiation game} is the game $\Red^\beta_{\lambda i}(G)_{\|v_0} = (\{\P, \C\}, S, (S_\P, S_\C), \Delta, \nu)_{\|v_0}$, where:
    \begin{itemize}
        \item the player $\P$ is called \emph{Prover}, and the player $\C$ \emph{Challenger};
        
        \item the states controlled by Prover are the states of $G$, i.e. $S_\P = V$;
        
        \item the states controlled by Challenger are the states of the form $[h c^\infty \rho]$, $(c, u)$ or $[h'v]$, where:
        \begin{itemize}
            \item $h c^\infty \rho$ is a $\lambda$-consistent $\beta$-punishment family,
            
            \item there exists a state $\rho_k \in V_i$ along the play $\rho$ such that $\rho_k u \in E$,
            
            \item and $h'v$ is a history such that $h'$ is a prefix of the history $hc$, and $\last(h') \in V_i$;
        \end{itemize}
        plus two additional states, written $\top$ and $\bot$;
        
        \item with the same notations, the set $\Delta$ contains the transitions of the form:
        \begin{itemize}
            \item $v [h c^\infty \rho]$ (Prover proposes a punishment family);
            
            \item $v \bot$ (Prover gives up);
            
            \item $[h c^\infty \rho] \top$ (Challenger accepts Prover's proposal);
            
            \item $[h c^\infty \rho] [h'v]$ (Challenger deviates before the punishing cycle --- \emph{pre-cycle deviation});
            
            \item $[h c^\infty \rho] (c, u)$ (Challenger deviates after the punishing cycle --- \emph{post-cycle deviation});
            
            \item $[h' v] v$ and $(c, u) u$ (Prover has now to propose a new play);
            
            \item $\top\top$ and $\bot\bot$ (the play is over);
        \end{itemize}
        
        \item given a history $H = H_0 \dots H_n \in \Hist \Red^\beta_{\lambda i}(G)$ that does not reach the state $\bot$, we write $\dH = h^{(1)} \dots h^{(n)}$ the history or play in $G$ defined by, for each $k$:
        \begin{itemize}
            \item if $H_{k-1}H_k = v [h c^\infty \rho]$, then $h^{(k)}$ is empty;
            
            \item if $H_{k-1}H_k = [h c^\infty \rho] \top$, then $h^{(k)} \dots h^{(n)} = h c^{\left|h^{(1)} \dots h^{(k-1)} h\right|^2} \rho$ (the number of times the cycle $c$ is repeated depends quadratically on the time);
            
            \item if $H_{k-1}H_k = [h c^\infty \rho] [h'v]$, then $h^{(k)} = h'$;
            
            \item if $H_{k-1}H_k = [h c^\infty \rho] (c, v)$, then $h^{(k)} = h c^{\left|h^{(1)} \dots h^{(k-1)} h\right|^2} h'$, where $h'$ is among the shortest histories such that $\Occ(h') \subseteq \Occ(\rho)$,  $\last(h') \in V_i$ and $\last(h') v \in E$;
            
            \item if $H_{k-1}H_k = [h' v] v$ or $(c, v) v$, then $h^{(k)}$ is empty;
        \end{itemize}
        
        and that definition is naturally extended to plays: for example, if $G$ is the game of Figure~\ref{fig_inf_memory} and if $\pi = a [ab^\infty a^\omega] [aa] a [ab^\infty a^\omega] (b, b) b [b^\infty a^\omega] \top^\omega$, then $\dpi = a \cdot ab^{2^2} \cdot a b^{7^2} a^\omega = a^2b^4ab^{49}a^\omega$;
        
        \item the payoff function $\nu$ is defined, for each play $\pi$, by $\nu_\C(\pi) = -\nu_\P(\pi) = +\infty$ if $\pi$ reaches the state $\bot$, and $\nu_\C(\pi) = -\nu_\P(\pi) = \mu_i(\dpi)$ otherwise.
    \end{itemize}
\end{defi}

\begin{exa}
    Figure~\ref{fig_reduced} illustrates a (small) part of the game $\Red_{\lambda \playcircle}^2(G)_{\|a}$, where $G$ is the game of Figure~\ref{fig_inf_memory}, and $\lambda(a) = \lambda(b) = 1$.
    Blue states are owned by Prover, orange ones by Challenger.
    When Prover proposes the punishment family $ab^\infty (a^3b^3)^\omega$, the function $\nu_\C$ interprets it as the play $ab^{|h|^2} (a^3b^3)^\omega$, where $h$ is the history that has already been constructed so far.
\end{exa}

\begin{rk}
    Reduced negotiation games are Borel, and are played on a finite graph.
\end{rk}

\subparagraph{\bf {Link with the negotiation function.}}

We will now prove that the reduced negotiation game captures the negotiation function, as do the abstract and concrete ones.
For that purpose, we first need the following key result.

\begin{lm} \label{lm_reduced_memoryless}
    In a reduced negotiation game, Prover has a memoryless optimal strategy.
\end{lm}

\begin{proof}
    This lemma is a consequence of Lemma~\ref{lm_memoryless}: the payoff function $\nu_\C$ is concave.
    Indeed, let $\xi$ be a shuffling of two plays $\pi$ and $\chi$.
    If either $\pi$ or $\chi$ reaches the state $\bot$ (in which case both do), then we immediately have $\nu_\C(\dxi) \leq \max\{\nu_\C(\dpi), \nu_\C(\dchi)\} = +\infty$.
    Otherwise, the play $\dxi$ is a shuffling of $\dpi$ and $\dchi$, and since mean-payoff objectives defined with a limit inferior are convex, we have $\nu_\C(\dxi) \leq \max\{\nu_\C(\dpi), \nu_\C(\dchi)\}$.
\end{proof}

This lemma enables us to prove that the reduced negotiation game is equivalent to the other negotiation games.

\begin{thm} \label{thm_reduced_game}
    There exists a polynomial $P_4$ such that for every mean-payoff game $G$, every requirement $\lambda$ with rational values, each player $i$ and each $v_0 \in V_i$, for every $\beta \geq P_4(\lv G \rv + \lv \lambda \rv)$, we have $\nego(\lambda)(v_0) = \val_\C \left( \Red^\beta_{\lambda i}(G)_{\|v_0} \right).$
\end{thm}

\begin{proof}
    For every mean-payoff game $G$ and every requirement $\lambda$, we assume:
    $$\beta \geq P_1 \left( P_2\left(\lv \left\{ \MP(c) ~|~ c \in \SC(G)\right\} \rv\right) \right)$$
    and for each $v \in V$:
    $$\beta \geq \lv \lambda(v) \rv + 3,$$
    which are indeed quantities that are bounded by a polynomial of $\lv G \rv +  \lv \lambda \rv$.
    
    \begin{itemize}
        \item \emph{First direction: $\nego(\lambda)(v_0) \geq \val_\C \left( \Red^\beta_{\lambda i}(G)_{\|v_0} \right).$}
        
        Let $\bsigma_{-i}$ be a strategy profile in $G$ that is $\lambda$-rational assuming a strategy $\sigma_i$, and let $x = \sup_{\sigma'_i} \mu_i(\< \bsigma_{-i}, \sigma'_i \>)$.
        We wish to prove that there exists a strategy $\tau_\P$ in the reduced negotiation game such that $\sup_{\tau_\C} \nu_\C(\< \btau \>) \leq x$.
        Thus, we will have proved that the quantity $\val_\C \left( \Red^\beta_{\lambda i}(G)_{\|v_0} \right)$ is smaller than or equal to every such $x$, and therefore smaller than or equal to $\nego(\lambda)(v_0)$.
        
        Let us define simultaneously the strategy $\tau_\P$ and a mapping $\phi: \Hist_\P\Red^\beta_{\lambda i}(G)_{\|v_0} \to \Hist G_{\|v_0}$, such that for each history $H$, the punishment family $\tau_\P(H)$ will be defined from the play $\< \bsigma_{\|\phi(H)} \>$.
        We guarantee inductively that if $H \in \Hist_\P\Red^\beta_{\lambda i}(G)_{\|v_0}$ is compatible with $\tau_\P$, then $\phi(H) \in \Hist G_{\|v_0}$ is compatible with $\bsigma_{-i}$.
        First, let us define $\phi(v_0) = v_0$.
        
        Let $H \in \Hist_\P \Red^\beta_{\lambda i}(G)_{\|v_0}$ be a history compatible with $\tau_\P$ as it has been defined so far, and such that $\phi(H)$ has already been defined.
        Let $\eta^0 = \< \bsigma_{\|\phi(H)} \>$.
        By induction hypothesis, the history $\phi(H)$ is compatible with $\bsigma_{-i}$, hence the play $\eta^0$ is $\lambda$-consistent, and satisfies $\mu_i(\eta^0) \leq x$.
        
        Let $\eta^0_{\leq \l}$ be the shortest prefix of $\eta^0$ that is not simple, i.e. such that there exists $k < \l$ with $\eta^0_k = \eta^0_\l$.
        If $\MP_i\left(\eta^0_{k+1} \dots \eta^0_\l\right) \leq x$, then we define $\tau_\P(H) = \left[\eta^0_{\leq k} \left( \eta^0_{k+1} \dots \eta^0_\l \right)^\infty \rho \right]$, where $\rho$ is a play such that $\Occ(\rho) = \Occ\left(\eta^0_{>\l}\right)$, that $\mu_i(\rho) \leq x$, and that $\lv \mu(\rho) \rv \leq \beta$.
        Such a play exists, because the polytope:
        $$Z = \left\{ \mu(\rho) ~\left|~ \begin{matrix}
            \forall j, \forall v \in V_j \cap \Occ(\eta^0), \mu_j(\rho) \geq \lambda(v), \\
            \mathrm{and~} \Occ(\rho) = \Occ\left(\eta^0_{>\l}\right)
        \end{matrix} \right. \right\}$$
        is nonempty (it contains $\eta^0_{>\l}$), and has at least one vertex $\bz$ with $z_i \leq x$ (because $\mu_i(\eta^0_{>\l}) \leq x$), which by Lemma~\ref{lm_conv_to_system} and Corollary~\ref{cor_size_vertices} has size $\lv\bz \rv \leq \beta$.
        
        Otherwise, if $\MP_i\left(\eta^0_{k+1} \dots \eta^0_\l\right) > x$, we define $\eta^1 = \eta^0_{\leq k} \eta^0_{>\l}$, and we iterate the process, which does necessarily terminate --- because $\mu_i(\eta^0) \leq x$.
        As a consequence, it effectively defines the proposal $\tau_\P(H) = \left[\eta^n_{\leq k} \left( \eta^n_{k+1} \dots \eta^n_\l \right)^\infty \rho \right]$, for some $n$.
        Then, for each prefix $hv$, we define $\phi\left(H \left[\eta^n_{\leq k} \left( \eta^n_{k+1} \dots \eta^n_\l \right)^\infty \rho \right] [hv] v\right) = \phi(H) \eta^0_{\leq m}$, where $\eta^0_{\leq m}$ is the prefix of $\eta^0$ of which $n$ simple cycles have been pulled out to obtain the prefix $h$ of $\eta^n$; and similarly, for each pair $(c, v)$, we define $\phi\left(H \left[\eta^n_{\leq k} \left( \eta^n_{k+1} \dots \eta^n_\l \right)^\infty \rho \right] (c, v) v\right) = \phi(H) \eta^0_{\leq m}$, where $\eta^0_{\leq m}$ is the prefix of $\eta^0$ from which $n$ simple cycles have been pulled out to obtain the shortest prefix $\eta^n_{\leq p}$ of $\eta^n$ such that $\eta^n_p \in V_i$ and $\eta^n_p v \in E$.
        
        Thus, the mapping $\phi$ is defined on every history compatible with $\tau_\P$, and the image of such a history is always a history compatible with $\bsigma_{-i}$.
        We define it arbitrarily on other histories.
        Note that for each history $H$, the history $\dH$ can be obtained from $\phi(H)$ by pulling out cycles $c$ satisfying $\MP_i(c) > x$, and adding cycles $d$ with $\MP_i(d) \leq x$.
        As a consequence, if $\MP_i(\phi(H)) \leq x$, then $\MP_i(\dH) \leq x$ --- and the same result is true when we naturally extend the mapping $\phi$ to plays.
        
        Let us now prove that $\sup_{\tau_\C} \nu_\C(\< \btau \>) \leq \sup_{\sigma'_i} \mu_i(\< \bsigma_{-i}, \sigma'_i \>)$.
        Let $\pi$ be a play compatible with $\tau_\P$:
        \begin{itemize}
            \item the state $\bot$ does not appear in $\pi$, because Prover's strategy does never use a transition to it.
            
            \item If $\pi$ has the form $\pi = H [h c^\infty \rho] \top^\omega$: then, we have $\nu_\C(\pi) = \mu_i(\rho) \leq x$.
            
            \item If $\pi$ is made of infinitely many deviations: the play $\phi(\pi)$ is compatible with $\bsigma_{-i}$, hence $\mu_i(\phi(\pi)) \leq x$; which implies $\mu_i(\dpi) \leq x$, i.e. $\nu_\C(\pi) \leq x$.
        \end{itemize}
        
        \item \emph{Second direction: $\nego(\lambda)(v_0) \leq \val_\C \left( \Red^\beta_{\lambda i}(G)_{\|v_0} \right).$}
        
        Let $\tau_\P$ be a memoryless strategy for Prover in the reduced negotiation game, and let $y = \sup_{\tau_\C} \nu_\C(\< \btau \>)$.
        We want to show that $\nego(\lambda)(v_0) \leq y$: by Lemma~\ref{lm_reduced_memoryless}, it will be enough to conclude.
        If $y = +\infty$, it is clear.
        Let us assume that $y \neq +\infty$.
        Then, we will define a strategy profile $\bsigma$, where $\bsigma_{-i}$ is $\lambda$-rational assuming $\sigma_i$, such that $\sup_{\sigma'_i} \mu_i(\< \bsigma_{-i}, \sigma'_i \>) \leq y$: we proceed inductively by defining the play $\< \bsigma_{\|hv} \>$ for each history $hv$ compatible with $\bsigma_{-i}$ such that $h$ is empty, or $\last(h) \in V_i$ and $v \neq \sigma_i(h)$.
        Such a history is called a \emph{bud history}.
        After other histories, the strategy profile can be defined arbitrarily.
        To that end, we construct a mapping $\psi$ which maps each bud history to a history $\psi(hv) \in \Hist_\P \Red_{\lambda i}^\beta(G)_{\|v_0}$ that is compatible with $\tau_\P$.
        This mapping will induce a definition of $\bsigma$: since $y \neq +\infty$, we have $\tau_\P(\psi(hv)) \neq \bot$: let then $[h' c^\infty \rho] = \tau_\P(\psi(hv))$.
        We then define $\< \bsigma_{\|hv} \> = h' c^{|hh'|^2} \rho$, which is a $\lambda$-consistent play since $h'c^\infty \rho$ is a $\lambda$-consistent punishment family, by definition of the reduced negotiation game.
        
        Let now $h_0 v$ be a bud history: we assume that $\bsigma$ has been defined on every prefix of $h_0$, but not on $h_0 v$ itself.
        If $h_0$ is empty, that is if $h_0v = v_0$, then we define $\psi(h_0v) = v_0$.
        Otherwise, let us write $h_0 = h_1 w h_2$, where $h_1 w$ is the longest prefix of $h_0$ that is a bud history --- that is, its longest prefix such that $\psi(h_1 w)$ has been defined, or its shortest prefix such that $wh_2$ is compatible with $\bsigma_{\|h_1 w}$.
        Let $H = \psi(h_1 w)$, and let $[h c^\infty \rho] = \tau_\P(H)$.
        We have defined $\< \bsigma_{\|h_1 w} \> = h c^{|h_1 h|^2} \rho$, and consequently, the history $w h_2$ is a prefix of that play.
        If it is a prefix of the history $hc$, then we define $\psi(h_0 v) = H [h c^\infty \rho] [w h_2 v] v$.
        Otherwise, we define $\psi(h_0 v) = H [h c^\infty \rho] (c, v) v$.
        
        Now, the strategy profile $\bsigma$ has been defined, and since all the punishment families proposed by Prover are $\lambda$-consistent, the strategy profile $\bsigma_{-i}$ is $\lambda$-rational assuming $\sigma_i$.
        Let $\eta$ be a play compatible with $\bsigma_{-i}$, and let us prove that $\mu_i(\eta) \leq y$.
        If $\eta$ has finitely many prefixes that are bud histories, then let $\eta_{\leq n}$ be the longest one: we have $\eta_{\geq n} = h c^{n+|h|} \rho$, where $[h c^\infty \rho] = \tau_\P(\psi(\eta_{\leq n}))$.
        Then, we have $\mu_i(\eta) = \mu_i(\rho) \leq y$.
        
        Now, if $\eta$ has infinitely many such prefixes, then there exists a unique play $\pi$ in the reduced negotiation game such that for any prefix $\eta_{\leq n}$ of $\eta$ that is a bud history, the history $\psi(\eta_{\leq n})$ is a prefix of $\pi$.
        Then, if $\pi$ contains finitely many post-cycle deviations, then there exist two indices $m$ and $n$ such that $\eta_{\geq m} = \dpi_{\geq n}$, hence $\mu_i(\eta) = \mu_i(\dpi) \leq y$.
        
        Finally, if $\pi$ contains infinitely many post-cycle deviations, i.e. infinitely many occurrences of a state of the form $(c, v)$, then let us choose such state that minimizes the quantity $\MP_i(c)$.
        The play $\eta$ has the form:
        $$\eta = h_0 c^{k_0^2} h_1 c^{k_1^2} h_2 \dots,$$
        where for each $n$, we have $k_n = \left|h_0 c^{k_0^2} \dots c^{k_{n-1}^2} h_n\right|$.
        Then, if we write $M = \max r_i$, we have:
        $$\MP_i\left(h_0 c^{k_0^2} \dots h_n c^{k_n^2}\right) \leq \frac{1}{k_n + k_n^2 |c| - 1} \left( k_n M + \left(k_n^2 |c| - 1\right) \MP_i(c) \right),$$
        which converges to $\MP_i(c)$ when $n$ tends to $+\infty$, hence $\mu_i(\eta) \leq \MP_i(c)$.
        Now, since $\tau_\P$ is memoryless, there exists a play of the form $H C^\omega$ that is compatible with it, and such that $(c, v) \in \Occ(C) \subseteq \Inf(\pi)$; and by definition of $y$, we have $\MP_i(\dC) = \nu_\C(HC^\omega) \leq y$.
        By minimality of $\MP_i(c)$, we have $\MP_i(\dC) = \MP_i(c)$, hence $\MP_i(c) \leq y$, and therefore $\mu_i(\eta) \leq y$. \qedhere
    \end{itemize}
\end{proof}

Thus, a given requirement $\lambda$ is an $\epsilon$-fixed point of $\nego$ if and only if for each $i$ and $v \in V_i$, there exists a memoryless strategy $\tau_\P$ in the game $\Red_{\lambda i}^\beta$, with $\beta = P_4(\lv G \rv + \lv \lambda \rv)$, such that $\sup_{\tau_\C} \nu_\C(\< \btau \>) \leq \lambda(v) + \epsilon$.
The reduced negotiation game has an exponential size, but it contains only $\card V$ states that are controlled by Prover: memoryless strategies for Prover are therefore objects of polynomial size.
Thus, such memoryless strategies constitute the third and last piece of our notion of witness.

        \section{Algorithm and complexity}

We are now in a position to define formally our notion of witness.

\begin{defi}[Witness]
    Let $I = (G_{\|v_0}, \bx, \by, \epsilon)$ be an instance of the $\epsilon$-SPE threshold problem.
    A \emph{witness} for $I$ is a tuple $\left(W, W', \bbalpha, \lambda, (\tau^v_\P)_v\right)$, where $W \subseteq W' \subseteq V$; $\bbalpha \in [0,1]^{\Pi \times \SC(W)}$;  $\lambda$ is a requirement; and each $\tau^v_\P$ is a memoryless strategy in the game $\Red^\beta_{\lambda i}(G)_{\|v}$, where $\beta = P_4(\lv G \rv + \lv \lambda \rv)$.
    A witness is \emph{valid} if:
    \begin{itemize}
        \item each strategy $\tau^v_\P$ satisfies the inequality $\sup_{\tau_\C} \nu_\C(\< \tau^v_\P, \tau_\C\>) \leq \lambda(v) + \epsilon$;
        
        \item the sets $W$ and $W'$ and the tuple of tuples $\bbalpha$ satisfy the hypotheses of Theorem~\ref{thm_constrained_existence_lambda_cons}.
    \end{itemize}
\end{defi}

\begin{rk}
    The sets $W$ and $W'$, as well as the tuple of strategies $(\tau^v_\P)_v$, have polynomial size.
    In order to bound the size of witnesses by a polynomial, we only have to bound $\lv \lambda \rv$ and $\left\lv \bbalpha \right\rv$.
\end{rk}

The $\epsilon$-SPE threshold problem will be $\NP$-easy if we show, first, that there exists a valid witness of polynomial size if and only if the instance is positive, and second, that the validity of a witness can be decided in polynomial time.
The former is a consequence of Lemma~\ref{lm_spe}, Theorem~\ref{thm_size_lambda}, Theorem~\ref{thm_constrained_existence_lambda_cons}, Theorem~\ref{thm_reduced_game}, and Lemma~\ref{lm_reduced_memoryless}:

\begin{lm} \label{lm_witness_existence}
    There exists a polynomial $P_5$ such that an instance $I$ of the $\epsilon$-SPE threshold problem admits a valid witness of size $P_5(\lv I \rv)$ if and only if it is a positive instance.
\end{lm}

Let us now tackle the latter.

\begin{lm} \label{lm_check_polynomial}
    Given an instance of the $\epsilon$-SPE threshold problem and a witness for it, deciding whether that witness is valid is $\PTime$-easy.
\end{lm}

\begin{proof}
    The validity of a witness is defined by two conditions.
    As regards the second one, all the hypotheses of Theorem~\ref{thm_constrained_existence_lambda_cons} can be checked in polynomial time with classical algorithms.
    Let us now show how the first condition can also be checked in polynomial time.
        
    Let $n = \card V$.
    Given a memoryless strategy $\tau^v_\P$ of Prover in a reduced negotiation game, one can construct in a time polynomial in $\lv \tau^v_\P \rv$ the graph $\Red_{\lambda i}^\beta(G)[\tau^v_\P]$, defined as the underlying graph of $\Red_{\lambda i}^\beta(G)$ where all the transitions that are not compatible with $\tau_\P^v$ have been omitted, as well as all the states that are, then, no longer accessible from the state $v$.
    That graph has indeed a polynomial size, because it is composed only of:
    \begin{itemize}
        \item at most $n$ vertices of the form $w \in V$;
        
        \item at most $n$ vertices of the form $\tau_\P(\cdot w)$ (either equal to $\bot$ or of the form $[h c^\infty \rho]$);
        
        \item at most $n^2$ vertices of the form $(c, w)$;
        
        \item at most $2n^2$ vertices of the form $[h'w']$, where $h'$ is a prefix of the history $hc$ for some punishment family $[hc^\infty\rho] = \tau_\P(\cdot w)$;
        
        \item possibly the state $\top$.
    \end{itemize}
    
    We call this connected graph the \emph{deviation graph}.
    Note that if among those vertices, there is the vertex $\bot$, then since the vertices that are not accessible have been removed, we have $\sup_{\tau_\C} \nu_\C(\< \btau \>) = +\infty$ and the problem can be solved immediately.
    In what follows, we assume that it is not the case, i.e. that for each $w$, the state $\tau_\P(\cdot w)$ has the form $[h c^\infty \rho]$.
    Deciding whether $\sup_{\tau_\C} \nu_\C(\< \btau \>) \leq \alpha$ is then equivalent to deciding whether there exists a path $\pi$, in that graph, such that $\nu_\C(\pi) > \alpha$.
    Such a play can have three forms.
    
    \begin{itemize}
        \item It can end in the state $\top$, i.e. with Challenger accepting Prover's proposal.
        The existence of such a play can be decided immediately, by checking whether in the deviation graph, there exists a vertex of the form $[h c^\infty \rho]$ with $\mu_i(hc^\infty \rho) > \alpha$.
        
        \item It can avoid the state $\top$, and comprise finitely many post-cycle deviations.
        This is the case if and only if there exists a cycle $C$ in the deviation graph, without post-cycle deviations, such that $\MP_i(\dC) > \alpha$.
        The existence of such a cycle can be decided in polynomial time with Karp's algorithm (see~\cite{DBLP:journals/dm/Karp78}).
        
        \item It can avoid the state $\top$, and comprise infinitely many post-cycle deviations.
        In that case, we have $\nu_\C(\pi) \leq \MP_i(c)$ for each state of the form $(c, w)$ appearing infinitely often along $\pi$; then, there exists a cycle $C$ in the deviation graph, such that every state of the form $(c, w)$ along $C$ satisfies $\MP_i(c) > \alpha$.
        Conversely, if such a cycle exists, then $\pi$ exists.
        The existence of such a cycle can be decided in polynomial time with Karp's algorithm.
    \end{itemize}
    
    Therefore, the existence of such a play is decidable in polynomial time.
\end{proof}

Thus, given an instance of the $\epsilon$-SPE threshold problem, a valid witness can be guessed and checked in polynomial time.
Since the $\epsilon$-SPE threshold problem has been proved to be $\NP$-hard in \cite{Concur}, we finally obtain the following theorem:

\begin{thm}\label{thm_np_complete}
    The $\epsilon$-SPE threshold problem in mean-payoff games is $\NP$-complete.
\end{thm}

\subparagraph{{\bf Acknowledgements}}{We wish to thank the anonymous reviewers for their useful comments, in particular for the question that led us to add Lemma~\ref{lm_np_hard} to this paper.}

\bibliography{bibli}

\newpage

%					+-----------------------+
%					!						!
%					!		APPENDICE		!
%					!						!
%					+-----------------------+

\appendix

The following appendices are providing the detailed proofs of all our results. They are not necessary to understand our results and are meant to provide full formalization and rigorous proofs. To improve readability, we have chosen to recall the statements that appeared in the main body of the paper before giving their detailed proofs in order to ease the work of the reader.

    \section{Proof of Lemma~\ref{lm_np_hard}} \label{pf_np_hard}

\begin{customlem}{\ref{lm_np_hard}}
    The SPE existence problem is $\NP$-hard.
\end{customlem}

\begin{proof}
    We proceed by reduction from the $\NP$-complete problem SAT.
    Let $\phi$ be a formula from propositional logic, written in conjunctive normal form.
    Let $X$ be the set of variables appearing in $\phi$.
    It has been proved in \cite{Concur} that one can construct in polynomial time a game $G^\phi_{\|v_0}$, whose player set is $\Pi = X \cup \{\S\}$ (the symbol $\S$ denotes a special player called \emph{Solver}), and in which there exists an SPE outcome $\rho$ with $\mu_\S(\rho) \geq 1$ if and only if $\phi$ is satisfiable.
    As an example, the game $G^\phi$ for $\phi = (x_1 \vee \neg x_1) \wedge \dots \wedge (x_6 \vee \neg x_6)$ is given in Figure~\ref{fig_Gphi}: the states named after a clause of $\phi$ are controlled by Solver, and the states of the form $\neg x$ are controlled by the player $x$.
    It is worth noting that in $G^\phi$, every play $\rho$ is such that either $\mu_\S(\rho) = 0$ or $\mu_\S(\rho) = 1$.
    
    \begin{figure}
        \centering
        \begin{tikzpicture}[->,>=latex,shorten >=1pt, scale=0.5, every node/.style={scale=0.5}, initial text={}]
    		\node[state, initial right] (C1) at (0:5) {$C_1$};
    		\node[state] (C11) at (30:6) {$x_1$};
    		\node[state] (C12) at (30:4) {$\neg x_1$};
    		\node[state] (C2) at (60:5) {$C_2$};
    		\node[state] (C21) at (90:6) {$x_2$};
    		\node[state] (C22) at (90:4) {$\neg x_2$};
    		\node[state] (C3) at (120:5) {$C_3$};
    		\node[state] (C31) at (150:6) {$x_3$};
    		\node[state] (C32) at (150:4) {$\neg x_3$};
    		\node[state] (C4) at (180:5) {$C_4$};
    		\node[state] (C41) at (210:6) {$x_4$};
    		\node[state] (C42) at (210:4) {$\neg x_4$};
    		\node[state] (C5) at (240:5) {$C_5$};
    		\node[state] (C51) at (270:6) {$x_5$};
    		\node[state] (C52) at (270:4) {$\neg x_5$};
    		\node[state] (C6) at (300:5) {$C_6$};
    		\node[state] (C61) at (330:6) {$x_6$};
    		\node[state] (C62) at (330:4) {$\neg x_6$};
    		\node[state] (b) at (0,0) {$\bot$};
    		
    		\path[->] (C1) edge node[right] {$\stackrel{x_1}{0} \stackrel{x_2}{1} \stackrel{x_3}{1} \stackrel{x_4}{1} \stackrel{x_5}{1} \stackrel{x_6}{1} \stackrel{\S}{1}$} (C11);
    		\path[->] (C1) edge (C12);
            \path[->] (C11) edge (C2);
            \path[->] (C12) edge (C2);
            \path[->] (C12) edge (b);
            \path[->] (C2) edge node[above right] {$\stackrel{x_1}{1} \stackrel{x_2}{0} \stackrel{x_3}{1} \stackrel{x_4}{1} \stackrel{x_5}{1} \stackrel{x_6}{1} \stackrel{\S}{1}$} (C21);
    		\path[->] (C2) edge (C22);
            \path[->] (C21) edge (C3);
            \path[->] (C22) edge (C3);
            \path[->] (C22) edge (b);
            \path[->] (C3) edge node[above left] {$\stackrel{x_1}{1} \stackrel{x_2}{1} \stackrel{x_3}{0} \stackrel{x_4}{1} \stackrel{x_5}{1} \stackrel{x_6}{1} \stackrel{\S}{1}$} (C31);
    		\path[->] (C3) edge (C32);
            \path[->] (C31) edge (C4);
            \path[->] (C32) edge (C4);
            \path[->] (C32) edge (b);
            \path[->] (C4) edge node[left] {$\stackrel{x_1}{1} \stackrel{x_2}{1} \stackrel{x_3}{1} \stackrel{x_4}{0} \stackrel{x_5}{1} \stackrel{x_6}{1} \stackrel{\S}{1}$} (C41);
    		\path[->] (C4) edge (C42);
            \path[->] (C41) edge (C5);
            \path[->] (C42) edge (C5);
            \path[->] (C42) edge (b);
            \path[->] (C5) edge node[below left] {$\stackrel{x_1}{1} \stackrel{x_2}{1} \stackrel{x_3}{1} \stackrel{x_4}{1} \stackrel{x_5}{0} \stackrel{x_6}{1} \stackrel{\S}{1}$} (C51);
    		\path[->] (C5) edge (C52);
            \path[->] (C51) edge (C6);
            \path[->] (C52) edge (C6);
            \path[->] (C52) edge (b);
            \path[->] (C6) edge node[below right] {$\stackrel{x_1}{1} \stackrel{x_2}{1} \stackrel{x_3}{1} \stackrel{x_4}{1} \stackrel{x_5}{1} \stackrel{x_6}{0} \stackrel{\S}{1}$} (C61);
    		\path[->] (C6) edge (C62);
            \path[->] (C61) edge (C1);
            \path[->] (C62) edge (C1);
            \path[->] (C62) edge (b);
            \path (b) edge[loop left] node[left] {$\stackrel{x_1}{1} \stackrel{x_2}{1} \stackrel{x_3}{1} \stackrel{x_4}{1} \stackrel{x_5}{1} \stackrel{x_6}{1} \stackrel{\S}{0}$} (b);
        \end{tikzpicture}
        \caption{The game $G^\phi$}
        \label{fig_Gphi}
    \end{figure}
    
    \begin{figure}
        \centering
		\begin{tikzpicture}[->,>=latex,shorten >=1pt, initial text={}, scale=0.7, every node/.style={scale=0.8}]
		\node[state, initial left] (a) at (0, 0) {$a$};
		\node[state, rectangle] (b) at (3, 0) {$b$};
		\node[state, rectangle] (c) at (6, 0) {$c$};
		\node (G) at (0, -2) {$G^\phi_{\|v_0}$};
		
		\path[bend left = 10] (a) edge node[above] {$\stackrel{\playcircle}{0}\stackrel{\Box}{3}$} (b);
		\path[bend left = 10] (b) edge node[below] {$\stackrel{\playcircle}{0}\stackrel{\Box}{3}$} (a);
        \path (b) edge (c);
		\path (c) edge [loop right] node {$\stackrel{\playcircle}{2}\stackrel{\Box}{2}$} (c);
		\path (a) edge (G);
		\end{tikzpicture}
		\caption{The game $H^\phi$}
		\label{fig_Hphi}
	\end{figure}
    
    Let us now construct a game $H^\phi$, according to Figure~\ref{fig_Hphi}.
    That game comprises all the players of $G^\phi$, plus two new ones, player $\Circle$ and player $\Box$.
    In the region corresponding to $G^\phi$, the rewards earned by those players are defined as follows: for each transition $uv$, we have $r_\Box(uv) = r_\playcircle(uv) = 1 - r_\S(uv)$.
    Thus, a play $\rho$ in that region is either such that $\mu_\S(\rho) = 1$ and $\mu_\Box(\rho) = \mu_\playcircle(\rho) = 0$, or that $\mu_\S(\rho) = 0$ and $\mu_\Box(\rho) = \mu_\playcircle(\rho) = 1$.
    The other rewards that are not written are not relevant for this proof and can be chosen arbitrarily.
    In the game $H^\phi_{\|a}$, an SPE exists if and only if the formula $\phi$ is satisfiable.
    
    \begin{itemize}
        \item Let us first assume that $\phi$ is satisfiable.
        Then, there exists an SPE outcome $\rho$ in $G^\phi_{\|v_0}$ such that $\mu_\S(\rho) = 1$, and therefore such that $\mu_\Box(\rho) = \mu_\playcircle(\rho) = 0$.
        Let $\btau$ be an SPE such that $\< \btau \> = \rho$.
        Then, the play $(ab)^\omega$ is an SPE outcome in the game $H^\phi_{\|a}$.
    
        Indeed, let us construct an SPE $\bsigma$ such that $(ab)^\omega = \< \bsigma \>$.
        First, for every history of the form $ha$, we define $\sigma_\playcircle(ha) = b$, and for every history of the form $hb$, we define $\sigma_\Box(hb) = a$.
        Second, for every history of the form $h a v_0$ that goes to the region corresponding to $G^\phi_{\|v_0}$, we define $\bsigma_{\|hav_0} = \btau$.
        Then, let $h$ be a history starting from $a$, and let us show that the substrategy profile $\bsigma_{\|h}$ is an NE.
        If $h$ traverses the state $v_0$, then $\bsigma_{\|h}$ is a substrategy profile of $\btau$, and therefore is an NE.
        If $h$ traverses the state $c$, then the result is immediate.
        Otherwise, i.e. if $h$ ends in $a$ or in $b$, then we have $\mu_\playcircle(\< \bsigma_{\|h} \>) = 0$, and player $\Circle$ cannot force a better payoff by going to the state $v_0$; and $\mu_\Box(\< \bsigma_{\|h} \>) = 3$, and if player $\Box$ chooses to go to the state $c$, he can get only the payoff $2$.
        
        \item Now, let us assume that there exists an SPE $\bsigma$ in $H^\phi_{\|a}$.
        Then, every substrategy profile $\bsigma_{\|hv_0}$ is an SPE, which can be considered as an SPE in the game $G^\phi_{\|v_0}$.
        If for some $h$, we have $\mu_\S(\< \bsigma_{\|hv_0} \>) = 1$, then the existence of such an SPE implies the satisfiability of $\phi$ by definition of $G^\phi$.
        Otherwise, if $\mu_\S(\< \bsigma_{\|hv_0} \>) = 0$, and therefore $\mu_\Box(\< \bsigma_{\|hv_0} \>) = \mu_\playcircle(\< \bsigma_{\|hv_0} \>) = 1$ for every history $hv_0$, then $\bsigma$ must satisfy $\mu_\playcircle(\< \bsigma_{\|ha} \>) \geq 1$ for every history $ha$ --- otherwise, player $\Circle$ would have a profitable deviation by going to the state $v_0$.
        That means that there is no history $ha$ from $a$ such that $\< \bsigma_{\|ha} \> = (ab)^\omega$.
        
        Therefore, there are two possibilities: either (1) there are infinitely many histories of the form $ha$ such that $\sigma_\playcircle(ha) = v_0$, or (2) there are infinitely many histories of the form $hb$ such that $\sigma_\Box(hb) = c$.
        The case (1) actually implies the case (2): whenever $\sigma_\playcircle(hba) = v_0$, we have $\sigma_\playcircle(hb) = c$, because player $\Box$ will have an incentive to go to the state $c$ and get the payoff $2$ before player $\Circle$ goes to $v_0$, which would give him only the payoff $1$.
        But then, being in the case (2) does also imply that we are not in the case (1): player $\Circle$ has no incentive to go to the state $v_0$, from which she can get only the payoff $1$, if player $\Box$ plans to go the state $c$, and to give her the payoff $2$, later.
        On the other hand, the case (2) implies the case (1): if after some history, player $\Circle$ does no longer plan to go to the state $v_0$, then player $\Box$ has no incentive to go to the state $c$, and staying in the cycle $(ab)^\omega$ is a profitable deviation from him.
        Contradiction: none of those cases is possible, and there exists necessarily $h$ such that $\mu_\S(\< \bsigma_{\|hv_0} \>) = 1$, i.e. there exists necessarily a valuation satisfying $\phi$.
    \end{itemize}
\end{proof}

    \section{Proof of Lemma~\ref{lm_conv_to_system}} \label{pf_conv_to_system}

\begin{customlem}{\ref{lm_conv_to_system}}
    There exists a polynomial $P_2$ such that, for each finite set $D$ and every finite subset $X \subseteq \R^D$, there exists a system of linear inequations $\Sigma$, such that $\Sol_\geq(\Sigma) = \Conv(X)$ and $\lv (\ba, b) \rv \leq P_2(\lv X \rv)$ for every $(\ba, b) \in \Sigma$.
\end{customlem}

\begin{proof}
    First, let us recall the notion of \emph{facet}: a facet of a polytope $P$ is a subset of $P$ of dimension $\dim P - 1$ and of the form $P \cap H$, where $H$ is a hyperplane defined by an equation $(\ba, b)$, such that $P \subseteq \Sol_\geq(\ba, b)$.

    Each facet of the polytope $\Conv(X)$ is of the form $\Conv\{\bx_1, \dots, \bx_n\}$, where $\bx_1, \dots, \bx_n$ are vertices of $X$ and $n \geq d = \card D$.
    Let $\Phi$ be the set of those facets.
    We can then write:
    $$\Conv(X) = \Sol_{\geq} \left\{ (\ba_F, b_F) ~|~ F \in \Phi\right\},$$
    where for each $F$, the equation $(\ba_F, b_F)$ defines the hyperplane to which the facet $F$ belongs.
    Let us now study the complexity of each of those equations (or inequations).
    
    For a given facet $F = \Conv\{\bx_1, \dots, \bx_n\}$, let us choose $d$ points $\by_1, \dots, \by_d \in \{\bx_1, \dots, \bx_n\}$ that are linearly independent.
    The equation $(\ba_F, b_F)$ has the points $\by_1, \dots, \by_d$ among its solutions, i.e. it satisfies:
    $$\forall i \in \{1, \dots, d\}, \sum_{j=1}^d a_{Fj} y_{ij} - b_F = 0.$$
    Those $d$ equalities can themselve be understood as equations on the pair $(\ba_F, b_F)$.
    Let us add a $(d+1)$-th equation: for some dimension $i_0$, we have $a_{Fi_0} \neq 0$, and by multiplying if needed by a nonzero factor, we can assume $a_{Fi_0} = 1$.
    By Lemma~\ref{lm_system_eq_to_point}, there exists a pair $(\ba, b)$ satisfying that system of $d+1$ equations and the inequality $\lv (\ba, b) \rv \leq P_1 \left( 2d + 5 + \max_i \left( \sum_j \lv x_{ij} \rv \right) \right) \leq P_1 (\lv X \rv)$.
    That pair is an equation of the hyperplane containing $F$.
    
    As a consequence, the polynomial $P_2 = P_1$ satisfies the desired inequalities.
\end{proof}

    \section{Proof of Theorem~\ref{thm_size_lambda}} \label{pf_size_lambda}

\begin{customthm}{\ref{thm_size_lambda}}
    There exists a polynomial $P_3$ such that for every mean-payoff game $G$, the least $\epsilon$-fixed point $\lambda^*$ of the negotiation function has size $\lv \lambda^* \rv \leq P_3(\lv G \rv + \lv \epsilon \rv)$.
\end{customthm}

\begin{proof}
    This proof uses the concrete negotiation game to construct the sets we wish to construct as polyhedra.
    We write $S = V \times 2^V \cup E \times 2^V$ the state space $\Conc_{\lambda i}(G)$, and $\nu_\C$ for Challenger's payoff function.

    \begin{itemize}
    \item \emph{Size of mean-payoffs of cycles}
    
    Without loss of generality, we assume that all the rewards $r_i(uv)$ for $i \in \Pi$ and $uv \in E$ are integers --- otherwise, we can multiply all the rewards by the least common denominator, compute the least fixed point in the resulting game, and finally divide each of its values by the least common denominator.
    
    Let now $i \in \Pi$ and $v \in V_i$, let $\lambda$ be a requirement, and let us consider the corresponding concrete negotiation game $\Conc_{\lambda i}(G)_{\|(v, \{v\})}$.
    It is shown in \cite{Concur} that Challenger has a memoryless optimal strategy.
    Let then $\tau_\C$ be a memoryless strategy for Challenger in $\Conc_{\lambda i}(G)_{\|v}$ and let $K$ be a strongly connected component of the graph $\Conc_{\lambda i}(G)[\tau_\C]$, i.e. the underlying graph of $\Conc_{\lambda i}(G)$ in which the transitions that are not compatible with $\tau_\C$ have been omitted.
    Let $c = c_1 \dots c_n$ be a simple cycle of $K$: necessarily, we have $n \leq \card S = \card V 2^{\card V}$.
    Therefore, we have, for each player $j$:
    \begin{align*}
        \lv \MP_j(\dc) \rv &= \left\lv \frac{1}{n} \sum_{k \in \Z/n\Z} r_j(\dc_k \dc_{k+1}) \right\rv \\
        &= 1 + \left\lceil \log_2\left( \left\vert \sum_{k \in \Z/n\Z} r_j(\dc_k \dc_{k+1}) \right\vert + 1 \right) \right\rceil + \lceil \log_2( n ) \rceil \\
        &\leq 1 + \left\lceil \log_2\left( \left\vert 2^{\card V} \sum_{e \in E} r_j(e) \right\vert + 1 \right) \right\rceil + \lceil \card V \log_2(\card V) \rceil \\
        &\leq 1 + \card V + \sum_{e \in E} \left\lceil \log_2\left( | r_j(e) | + 1 \right) \right\rceil + (\card V)^2 \\
        &\leq (\card E)^3 + \sum_{e \in E} \left\lceil \log_2\left( | r_j(e) | + 1 \right) \right\rceil \\
        &\leq \lv r_j \rv^3 \leq \lv G \rv^3.
    \end{align*}

    \item \emph{Feasible payoff vectors in a strongly connected component}
    
    Let us now define the polytope:
    $$F_K = \dseal \left(\underset{c \in \SC(K)}{\Conv} \MP(\dc) \right).$$
    By Lemma~\ref{lm_dseal}, the polytope $F_K$ is exactly the set of the payoff vectors of plays in the strongly connected component $K$.
    By Lemma~\ref{lm_conv_to_system}, the polytope:
    $$\underset{c \in \SC(K)}{\Conv} \MP(\dc)$$
    is the solution set of a system of inequations whose sizes are bounded by a polynomial function of $\lv \{\MP(\dc) ~|~ c \in \SC(K) \} \rv$, i.e., according to the previous point, of $\lv G \rv$.
    By Lemma~\ref{lm_dseal_size}, this is also the case for the polytope $F_K$.

    \item \emph{Consistent payoff vectors in a strongly connected component}
    
    If $K$ contains no deviation, then there is a set $M_K \subseteq V$ such that all the states of $K$ have the form $(\cdot, M_K)$.
    If $K$ contains a deviation, we define $M_K = \emptyset$.
    Given a requirement $\lambda$, we can now define a new polytope:
    $$C_{K \lambda} = \left\{\bx \in F_K ~|~ \forall j, \forall u \in V_i \cap M_K, x_i \geq \lambda(u)\right\}.$$
    That polytope is the set of the tuples $\bx$ such that there exists a play $\pi$ in $K$ realizing $\mu(\dpi) = \bx$, \emph{and} $\nu_\C(\pi) \neq +\infty$.
    Note that we have:
    $$\nego(\lambda)(v) = \sup_{\tau_\C} \inf_K \inf_{\bx \in C_{K\lambda}} x_i.$$
    The inequations defining $C_{K\lambda}$ are the same as $F_K$, plus the inequations of the form $x_i \geq \lambda(u)$.

    \item \emph{Union, intersection and product}
    
    Let us now work in the space $\R^{V \times \Pi}$.
    Given a requirement $\lambda$, we define the set:
    $$X_\lambda = \prod_{i \in \Pi, v \in V_i} \bigcap_{\tau_\C \in \ML\left( \Conc_{\lambda i}(G)_{\|(v, \{v\})}\right)} \bigcup_K C_{K\lambda}^\uparrow,$$
    where $Y \mapsto Y^\uparrow$ is the upward closure operation, i.e. $Y^\uparrow = \{z ~|~ \exists y \in Y, z \geq y\}$.
    Then, for every tuple of tuples $\bbx \in \R^{V \times \Pi}$, we have $\bbx \in X_\lambda$ if and only if for each $i$ and $v \in V_i$, for every memoryless strategy of Challenger, there exists a play $\pi$ compatible with $\tau_\C$ satisfying $\nu_\C(\pi) \neq +\infty$ and $\mu(\dpi) \leq \bx_v$.
    Therefore, for each $i$ and $v \in V_i$, we have $\nego(\lambda)(v) = \inf\left\{x_{vi} ~\left|~ \bbx \in X_\lambda\right.\right\}$.
    
    In terms of inequations, the set $X_\lambda$ is a union of polyhedra which are all defined by inequations that are inequations defining some $C_{K\lambda}$, padded with $0$ to fit with the dimension change.

    \item \emph{$\epsilon$-fixed points}
    
    To each tuple of tuples $\bbx$, we associate the requirement $\lambda_{\bbx}$ defined, for each $i \in \Pi$ and $v \in V_i$, by $\lambda_{\bbx}(v) = x_{vi} - \epsilon$.
    Now, let us consider the set:
    $$X = \left\{ \bbx \in \R^{V \times \Pi} ~\left|~ \bbx \in X_{\lambda_{\bbx}} \right.\right\}.$$
    
    In terms of inequations, the set $X$ is a union of polyhedra defined by the same inequations as $X_\lambda$, but where the inequations of the form $x_{ui} \geq \lambda(v)$ are replaced by equations of the form $x_{ui} \geq x_{vi}$, of size $2 + 4\card V \card\Pi$.
    
    Let $\lambda$ be a requirement.
    Then, it is an $\epsilon$-fixed point of the negotiation function if and only if $\lambda = \lambda_{\bbx}$ for some $\bbx \in X$.
    
    \begin{itemize}
        \item Indeed, if there exists $\bbx \in X$ such that $\lambda = \lambda_{\bbx}$, then we have $\bbx \in X_{\lambda_{\bbx}}$, and by the previous point, for each $i$ and $v \in V_i$, we have $\nego(\lambda_{\bbx})(v) \leq x_{vi} = \lambda_{\bbx}(v) + \epsilon$.
        Therefore, $\lambda$ is an $\epsilon$-fixed point of the negotiation function.
        
        \item Conversely, if for each $i$ and $v \in V_i$, we have $\nego(\lambda)(v) \leq \lambda(v) + \epsilon$, according to the previous point and since the set $X_{\lambda}$ is closed, there exists a tuple of tuples $\bbx^{(v)} \in X_{\lambda}$ such that $x^{(v)}_{vi} = \lambda(v) + \epsilon$.
        Then, since $X_{\lambda}$ is defined as a cartesian product over $v$, the tuple of tuples $\bbx = \left(\bx^{(v)}_v\right)_v$ does also belong to $X_{\lambda}$, and satisfies $\lambda = \lambda_{\bbx}$.
    \end{itemize}
    
    Then, in particular, the least $\epsilon$-fixed point $\lambda^*$ is the (unique) minimal element in the set $\left\{\lambda_{\bbx} ~\left|~ \bbx \in X\right.\right\}$.
    The set $X$ is itself a union of polyhedra: consequently, the linear mapping $\bbx \mapsto \sum_v \lambda_{\bbx}(v)$ has its minimum over $X$ on some vertex $\bbx$ of one of those polyhedra, which is therefore such that $\lambda^* = \lambda_{\bbx}$.
    By Corollary~\ref{cor_size_vertices}, that vertex has size bounded by a polynomial function of the maximal size of the inequations defining $X$, and therefore a polynomial function of $\lv G \rv + \lv \epsilon \rv$.
    \end{itemize}
\end{proof}

\end{document}